\documentclass[reqno, 12pt]{amsart}

\usepackage{amsmath, amsthm, amssymb, mathtools, microtype}
\usepackage[dvipsnames]{xcolor}
\usepackage[a4paper,margin=2cm]{geometry} 
\usepackage[utf8]{inputenc}
\usepackage[all]{xy}
\usepackage[hidelinks]{hyperref}
\usepackage[mathscr]{euscript}
\numberwithin{equation}{section}

\theoremstyle{plain}
\newtheorem{theorem}{Theorem}
\newtheorem{proposition}[theorem]{Proposition}
\newtheorem{lemma}[theorem]{Lemma}
\newtheorem{corollary}[theorem]{Corollary}
\theoremstyle{definition}

\newtheorem{definition}[theorem]{Definition}

\theoremstyle{remark}
\newtheorem{remark}[theorem]{Remark}

\DeclareMathOperator*{\res}{Res}
\DeclareMathOperator{\diag}{diag}

\def\Z{\mathbb{Z}}	
\def\C{\mathbb{C}}	
\def\R{\mathbb{R}}	
\def\cA{\mathcal{A}}

\def\cD{\mathcal{D}}
\def\cE{\mathcal{E}}
\def\cL{\mathcal{L}}
\def\cN{\mathcal{N}}
\def\cS{\mathcal{S}}
\def\cT{\mathcal{T}}
\def\cP{\mathcal{P}}
\def\cQ{\mathcal{Q}}
\def\cV{\mathcal{V}}
\def\cW{\mathcal{W}}
\def\cZ{\mathcal{Z}}
\def\cf{\mathfrak{f}}
\def\ch{\mathfrak{h}}
\def\cU{\mathcal{U}}
\def\i{\mathsf{i}}
\def\I{\mathsf{i}}
\def\hS{\hat{S}}
\DeclareMathOperator{\End}{End}

\def\unity{1\!\!1}
\def\KPt{\mathsf{t}}
\def\KPp{\mathsf{p}}
\def\Log{\mathop{\mathrm{Log}}}

\begin{document}

\title[Hypermaps and constrained KP]{Enumeration of hypermaps and Hirota equations for extended rationally constrained KP}

\author{G. Carlet}
\address{G.~C.: Institut de Mathématiques de Bourgogne, UMR 5584 CNRS, Université Bourgogne Franche-Comté, F-2100 Dijon, France}
\email{Guido.Carlet@u-bourgogne.fr}

\author{J. van de Leur}
\address{J.~v.d.L: Mathematical Institute, Utrecht University,  P.O. Box 80010, 3508 TA Utrecht, The Netherlands}
\email{J.W.vandeLeur@uu.nl}

\author{H. Posthuma}
\address{H.~P.: Korteweg-de Vries Institute for Mathematics, University of Amsterdam, Postbus 94248, 1090 GE Amsterdam, The Netherlands}
\email{H.B.Posthuma@uva.nl}	

\author{S. Shadrin}
\address{S.~S.: Korteweg-de Vries Institute for Mathematics, University of Amsterdam, Postbus 94248, 1090 GE Amsterdam, The Netherlands}
\email{S.Shadrin@uva.nl}	

\begin{abstract} We consider the Hurwitz Dubrovin--Frobenius manifold structure on the space of meromorphic functions on the Riemann sphere with exactly two poles, one simple and one of  arbitrary order. We prove that the all genera partition function (also known as the total descendant potential) associated with this Dubrovin--Frobenius manifold is a tau function of a rational reduction of the Kadomtsev--Petviashvili hierarchy. This statement was conjectured by Liu, Zhang, and Zhou. We also provide a partial enumerative meaning for this partition function associating one particular set of times with enumeration of rooted hypermaps.
\end{abstract}

\maketitle

\tableofcontents

\section*{Introduction}

This paper is a continuation to~\cite{carletHigherGeneraCatalan2021}. We consider the Hurwitz  space $H_{0;1,N-1}$ of rational functions with two poles, one of order $1$ and one of order $N-1$, for $N\geq 2$. As shown by Dubrovin in~\cite{dub96}, this space carries the structure of semi-simple Dubrovin--Frobenius manifold of rank $N$. This family of Frobenius manifolds has also been recently studied from a different point of view in~\cite{Arsie2021DubrovinFrobenius}.

There are two equivalent constructions of the so-called all genera partition function (also known as total descendant potential) associated with a semi-simple Dubrovin--Frobenius manifold, given respectively in~\cite{Giv-Semisimple,Giv2001} and in~\cite{DZ01} (their equivalence is proved in~\cite{DZ01}). 
Until now the combinatorial or enumerative meaning of the partition functions associated to $H_{0;1,N-1}$ was not known, except for the case $N=2$ covered in~\cite{carletHigherGeneraCatalan2021}. 
We prove in this paper that a restriction of this partition function to one set of times gives a generating function for the enumeration of rooted hypermaps on compact two-dimensional surfaces. 

Dubrovin and Zhang proved in \emph{op.~cit.} that the partition function of any Dubrovin--Frobenius manifold is a tau function of an infinite-dimensional integrable system, the so-called Dubrovin--Zhang hierarchy (see also~\cite{BDS2012}). Liu, Zhang, and Zhou conjectured in~\cite{liuCentralInvariantsConstrained2015} that the Dubrovin--Zhang integrable systems associated with the family of Dubrovin--Frobenius manifolds on $H_{0;1,N-1}$, $N\geq 2$, are certain extensions of the rationally $(N-1)$-constrained Kadomtsev--Petviashvili (KP) hierarchies. In this paper we prove a version of their conjecture. Namely, we do prove that the partition function associated to $H_{0;1,N-1}$, $N\geq 2$, is a tau function of the extended rationally constrained KP hierarchies, using an approach to construct the integrable systems due to Givental--Milanov--Tseng~\cite{Giv-An,GivMil,MilTse}. Note that the equivalence of these two approaches to the construction of integrable systems is not known. The first few steps towards their identification are outlined in~\cite{CLPS-Lax}.

Extending the method of Givental--Milanov--Tseng, we first construct Hirota bilinear equations for the partition function of $H_{0;1,N-1}$, $N\geq 2$, and then we derive from the Hirota bilinear equations the Lax representation of the corresponding integrable system, generalizing the methods developed in~\cite{Mil2007, CvdL2013}. Finally, in the Lax form we recognize the rationally $(N-1)$-constrained KP hierarchy, as predicted by Liu, Zhang, and Zhou in~\cite{liuCentralInvariantsConstrained2015}.

\subsection*{Organization of the paper} To keep this paper reasonably short and technical we expect the reader to be familiar with the theory of Dubrovin--Frobenius manifolds as well as with the Givental theory. We will focus on the particular computations necessary to generalize the results of~\cite{carletHigherGeneraCatalan2021}, to which we refer for a survey of the relevant definitions in the same notation as we use here. 

In Section~\ref{sec:FrobeniusManifold} we recall the definition and the basic structures, like flat metric, flat coordinates, the unit and the Euler vector fields, etc. of the Dubrovin--Frobenius manifold $H_{0;1,N-1}$. In the further computations of the various structures of this Dubrovin--Frobenius manifold we restrict to a special point that we use as the base point for the expansion of the total descendant potential. In particular, we recall there the choice of calibration proposed by Liu, Zhang, and Zhou. 

In Section~\ref{sec:PartitionFunction} we recall the Givental formulas for the total descendant and ancestor potentials (expanded at the special point) and prove that the restriction of the total descendant potential to one set of times gives a generating function for the enumeration of rooted hypermaps, and thus a hypergeometric KP tau function by itself (see Theorem~\ref{thm:Hypermaps} and Corollary~\ref{cor:KP}).

In Section~\ref{sec:PeriodVectors} we discuss the geometry behind the solutions of Dubrovin's Fuchsian system associated to the Dubrovin--Frobenius manifold $H_{0;1,N-1}$ and choose a suitable (sub)orbit of its monodromy group for further analysis. In particular, we discuss the period vectors associated to this (sub)orbit. Using these period vectors, we construct the vertex operators and discuss their asymptotics and conjugation properties (needed in the next section for the fomulation and analysis of the Hirota quadratic equations).

In Section~\ref{sec:HirotaAncestor} we formulate and prove the Hirota quadratic equations for the total ancestor  and the total descendant potentials, see Theorem~\ref{thm:HirotaForAncestors} and Theorem~\ref{thm:DescHirota}, respectively. Note that we consider the expansions of the both potentials at the special point of $H_{0;1,N-1}$. At the end of the section we finally present the Hirota quadratic equations for the descendant potential in a very explicit form, see Corollary~\ref{cor:explicitform}.

In Section~\ref{sec:LaxFormulation} we analyse the explicit form of the Hirota quadratic equations and derive the Lax formulation of the associated integrable system, which we recognise to be the rationally $(N-1)$-constrained KP hierarchy, see Proposition~\ref{prop:RRKP} and Proposition~\ref{prop:KricheverRR}. In Proposition~\ref{logprop} we provide the Lax representation of the extra chain of flows that provide the mentioned extension of the well-known rational reductions of the KP hierarchy. It has to be noted, however, that certain expected properties of the Lax pairs associated to the extra flows could not be proved with the currently available methods, see Remark~\ref{extraconj} for the related conjecture. 

\subsection*{Acknowledgments} 
This work is supported by the EIPHI Graduate School (contract ANR-17-EURE-0002).
H.~P. and S.~S. were supported by the Netherlands Organization for Scientific Research.

\section{The Dubrovin--Frobenius manifold}
\label{sec:FrobeniusManifold}

We recall the definition of the family of Dubrovin--Frobenius manifolds under consideration in terms of their superpotential following the construction on Hurwitz spaces in~\cite[Lecture 5]{dub96}. This family of Dubrovin--Frobenius manifolds was also studied in detail in~\cite[Section 8]{db19}. The $N=2$ case is discussed in detail in~\cite{carletHigherGeneraCatalan2021}
and some analysis of the case $N=3$ is developed in~\cite[Example 5.5]{dub96}.

Most of our analysis throughout the paper is performed in the neighborhood of the special point $t_{sp}$, defined below. After a short survey of the basic structures of the Dubrovin--Frobenius manifold, we switch to the computations at $t_{sp}$. 

\subsection{Definition and basic structures} 

Let $M$ be the Hurwitz space $H_{0;1,N-1}$, i.e. the space of meromorphic functions on the Riemann sphere $\C_\infty$ with two poles of order $1$ and $N-1$ respectively and simple finite ramification points, modulo automorphisms of $\C_\infty$.  
The space $\tilde{M}$ of meromorphic functions on $\C_\infty$  with ramification profile $(N-1,1)$ over $\infty$, modulo automorphisms, can be identified with $\C^{N-1}\times\C^*$ via the {\it superpotential}
\begin{equation}
f(p) = p^{N-1} + a_{2} p^{N-3} + \dots + a_{N-1} + \frac{a_N}{p-a_1},
\end{equation}
so $M$ is identified with an open subset of $\tilde{M}\simeq\C^{N-1}\times\C^*$, given by the functions $f(p)$ which have $N$ distinct finite simple critical points. 
Denoting $\Delta_{crit}$ the corresponding degenerate subset of $\tilde{M}$, we have that $M=\tilde{M}\setminus  \Delta_{crit}$. 

In the neighbourhood of any point $a \in M$, the critical values $u^i = f(\tilde{p}_i)$,  corresponding to the $N$ distinct critical points $\tilde{p}_i$ with $f'(\tilde{p}_i)=0$,  $i=1, \dots, N$, define local coordinates. The assumption of simple ramification at the finite ramification points means that $f ''(\tilde{p}_i)\not= 0$. 
Denote by $\hat{M}$ the subset of $M$ where $f$ has $N$ distinct critical values. 

The commutative associative product $\cdot$ on the tangent spaces to $M$ is defined by declaring the coordinates $u^i$ to be canonical, i.e., the coordinate tangent fields are idempotents,
\begin{equation}
\frac{\partial }{\partial u^i} \cdot \frac{\partial }{\partial u^j} = \delta_{ij} \frac{\partial }{\partial u^i}.  
\end{equation}

The action of the affine group $f \mapsto \alpha f +\beta$ on $H_{0;1,N-1}$ induces the unit and the Euler vector fields which are explicitly given by 
\begin{equation}
e = \frac{\partial }{\partial a_{N-1}}, \qquad 
E = \sum_i  \frac{i}{N-1} a_i \frac{\partial }{\partial a_i}.
\end{equation}
By direct computation we have
\begin{equation} \label{eq:ActionEonSuperPot}
E(f) = f-\frac 1{N-1} p \frac{\partial f}{\partial p}, \qquad e(f)=1,
\end{equation}	
from which it follows that in canonical coordinates 
\begin{equation} 
e = \sum_i \frac{\partial }{\partial u^i}, \qquad 
E = \sum_i u^i \frac{\partial }{\partial u^i}.
\end{equation}

The metric is given by 
\begin{equation}
( \partial' , \partial'' ) = - (\res_{p=\infty} + \res_{p=a_1} ) \frac{\partial' f \ \partial'' f }{\partial_p f} \ dp .
\end{equation}
The flatness of the metric can be directly proved by introducing flat coordinates $t^1,\dots,t^N$ given by $t^1 = a_1$ and
\begin{equation}
t^\alpha  = - \frac{N-1}{\alpha-1} \res_{p=\infty} f^{\frac{\alpha-1}{N-1}}\ dp, \qquad  \alpha=2,\dots,N, 
\end{equation}
which in particular implies $t^N = a_N$. The metric in these coordinates is given by 
\begin{equation} \label{eq:Eta-Definition}
	\eta_{\alpha\beta} = \begin{cases}
		\frac{\delta_{\alpha+\beta,N+1}}{N-1} & 2\leq \alpha \leq N-1, \\
		\delta_{\alpha+\beta,N+1} & \alpha=1,N.
	\end{cases} 
\end{equation}

Using~\eqref{eq:ActionEonSuperPot} it is easy to see that the unit and the Euler vector fields in flat coordinates are given by
\begin{align}
e= \frac{\partial}{\partial t^{N-1}}, \qquad	E = \sum_{\alpha=1}^{N} \frac{\alpha}{N-1} t^\alpha\frac{\partial}{\partial t^\alpha }.
\end{align}

The \emph{charge} $d$ of this Dubrovin--Frobenius manifold is equal to 
\begin{equation}
d = \frac{N-3}{N-1}.	
\end{equation}

Recall also that the matrices $\mu$ and $\cU$ are respectively defined as 
\begin{equation}
\mu = \frac{2-d}2 - \nabla E = \diag\left(\frac{N-1}{2(N-1)}, \frac{N-3}{2(N-1)}, \dots, \frac{1-N}{2(N-1)}\right), \qquad 
\cU = E \cdot
\end{equation}
and are respectively skew-symmetric and symmetric w.r.t. the inner product $\eta$
\begin{equation}  
\eta \mu  = -  \mu\eta, \qquad 
\eta\,  \cU   = \cU^T\!\eta.
\end{equation}

\subsection{Computations at the special point} 
\label{sec:SpecialPoint}

We fix a special point $t_{sp}\coloneqq \{a_i=\delta_{i,N}\}$, where the superpotential takes the form $f(p,t_{sp})=p^{N-1}+p^{-1}$. Our goal is to compute some of the structures associated with the Dubrovin--Frobenius manifold $M$  at this point. 

In all computations below it is sufficient to use the fact that in the neighborhood of the special point (that is, in the coordinates $t^\alpha = \tilde t^\alpha+\delta^{\alpha,N}$)
\begin{equation}
	f(p,t) = f(p,t_{sp})+ \sum_{\alpha=2}^{N} \tilde t^\alpha p^{N-1-\alpha} + \tilde t^1 p^{-2} +o(\tilde t^1,\dots,\tilde t^N).
\end{equation}

The critical points of $f(p,t_{sp})$ are $c_i = (N-1)^{-1/N}\exp(2\pi\i \cdot i/N)$, $i=1,\dots,N$. Therefore, the values of the canonical coordinates at $t_{sp}$ are $u^i = f(c_i,t_{sp})$, $i=1,\dots,N$, and 
\begin{equation}
	\frac{\partial u^i}{\partial t^\alpha}\Big|_{t_{sp}} = 
	\begin{cases}
		c_i^{N-1-\alpha}, & \alpha = 2,\dots,N, \\
		c_i^{-2}=(N-1)c_i^{N-2}, & \alpha = 1. 
	\end{cases}
\end{equation}
Note that 
\begin{equation}
	\delta_{ij}\Delta_i^{-1} |_{t_{sp}} = \left(\frac{\partial}{\partial u^i},\frac{\partial}{\partial u^j}\right)\Big|_{t_{sp}} = \delta_{ij} N^{-1} (N-1)^{-3/N} \exp(2\pi\i\cdot 3i/N).
\end{equation}
Fix
\begin{equation}
	\Delta_i^{1/2} |_{t_{sp}} = N^{1/2} (N-1)^{3/2N} \exp(-2\pi\i\cdot 3i/2N).
\end{equation}
Then
\begin{align}
	\Psi^i_\alpha |_{t_{sp}} & = \Delta_i^{-1/2}\frac{\partial u^i}{\partial t^\alpha}\Big|_{t_{sp}}
	\\ \notag 
	& = \begin{cases}
	N^{-1/2} (N-1)^{(-2N-1+2\alpha)/2N} \exp(2\pi\i \cdot i(2N+1-2\alpha)/2N), & \alpha = 2,\dots,N, \\
	N^{-1/2} (N-1)^{1/2N} \exp(2\pi\i \cdot i(-1)/2N), & \alpha =1.	
	\end{cases}
\end{align}

\subsubsection{The $R$ matrix} \label{sec:R-matrix}

The Givental $R$ matrix given by $R(z)=\sum_{\ell=0}^\infty R_\ell z^\ell$ is defined by the equations $R_0=\mathrm{Id}$ and \begin{align}
	[\Psi^{-1}R_{m+1} \Psi, \cU] = (m+\mu) \Psi^{-1}R_m\Psi.
\end{align}	
Knowing $\Psi|_{t_{sp}}$ and $\cU|_{t_{sp}}$ this equation allows to fully reconstruct $R(z)|_{t_{sp}}$, but we won't need the explicit form of the $R$-matrix.

\subsubsection{Calibration} \label{sec:Calibration} 

We recall the choice of calibration proposed in~\cite[Section 5]{liuCentralInvariantsConstrained2015}. Let $f(p)=f(p,a)$. We have, for $m \geq 0$: 
\begin{align}
	\theta_{\alpha,m} & \coloneqq -\frac{1}{(N-1)\prod_{k=0}^m \left(k+ \frac{N-\alpha}{N-1}\right)}\res_{p=\infty} f(p)^{m+\frac{N-\alpha}{N-1}} dp,& \alpha=2,\dots,N-1,\\
	\theta_{1,m} & \coloneqq -\frac{1}{(m+1)!}\res_{p=\infty} f(p)^{m+1}dp , & 
	\\	
	\theta_{N,m} & \coloneqq \frac{N}{(N-1)m!}\res_{p=a_1}  f(p)^m(\widetilde\log f(p) - \ch(m)) dp .& 
\end{align}
In the last expression $\widetilde\log f(p)$ is defined as a series near $p=a_1$
\begin{align}
	\widetilde\log f(p) \coloneqq \frac{N-1}{N} \left(\log a_N + A + \frac 1{N-1} B\right),
\end{align}
where $A$ is a formal power series in $p-a_1$ obtained by the expansion of $\log\left(\frac{p-a_1}{a_N} f(p,a)\right)$ at $p\sim a_1$, and $B$ is a formal power series in $(p-a_1)^{-1}$ obtained by the expansion of $\log\left((p-a_1)^{1-N}f(p,a)\right)$ at $(p-a_1)^{-1}\sim 0$.

The $S$-matrix is then defined as 
\begin{align}\label{eq:S-matrix-general-def}
(S_m)^\alpha_\beta\coloneqq \eta^{\alpha\gamma} \frac{\partial \theta_{\beta,m}}{\partial t^\gamma} . 	
\end{align}

\subsubsection{The $S$ matrix}

From the previous formula and the choice of calibration 
of Liu--Zhang--Zhou we obtain the following formula for the $S$-matrix at $t_{sp}$. Letting $F\coloneqq p^{N-1}+p^{-1}=f(p,t_{sp})$ we have:
\begin{align} \label{eq:S-at-tsp}
	(S_m)^{\alpha}_\beta = \begin{cases}
		-\frac{1}{\prod_{k=0}^{m-1} (k+\frac{N-\beta}{N-1})} \res\limits_{p=\infty} p^{\alpha-2} F^{m-(\beta-1)/(N-1)} dp, & 2\leq \alpha,\beta\leq N-1,
		\\ 
		-\frac{1}{(N-1)\prod_{k=0}^{m-1} (k+\frac{N-\beta}{N-1})} \res\limits_{p=\infty} p^{-1} F^{m-(\beta-1)/(N-1)} dp, & \alpha=1, 2\leq \beta\leq N-1, 
		\\
		-\frac{1}{(N-1)\prod_{k=0}^{m-1} (k+\frac{N-\beta}{N-1})} \res\limits_{p=\infty} p^{-2} F^{m-(\beta-1)/(N-1)} dp,  & \alpha=N, 2\leq \beta\leq N-1, 
		\\
		-\frac{N-1}{m!}\res\limits_{p=\infty} p^{\alpha-2} F^{m} dp,
		& 2\leq \alpha\leq N-1, \beta=1,
		\\
		-\frac{1}{m!}\res\limits_{p=\infty} p^{-1}  F^{m} dp,
		& \alpha=1, \beta=1,
		\\
		-\frac{1}{m!}\res\limits_{p=\infty} p^{-2}  F^{m} dp,
		& \alpha=N, \beta=1.
		\end{cases}
\end{align}
The case of $\beta=N$ is a bit more subtle. We have to use four special series there. Let 
\begin{align}
A & \coloneqq \log(1+p^N) \text{ expanded in } p, &
B & \coloneqq \log(1+p^{-N}) \text{ expanded in } p^{-1}, \\ \notag 
C & \coloneqq (1+p^{N})^{-1} \text{ expanded in } p, &
D&  \coloneqq (1+p^{-N})^{-1} \text{ expanded in } p^{-1}.
\end{align}
Then we have:
\begin{align} \label{eq:S-at-tsp=beta-N}
	(S_m)^{\alpha}_N = \begin{cases}		
\frac{N}{m!} \res\limits_{p=0} p^{\alpha-2}  m F^{m-1}(\frac{N-1}{N}A+\frac 1N B - \ch(m))
& 
\\
\qquad +\frac{N}{m!} \res\limits_{p=0}  F^{m}(\frac{N-1}{N}Cp^{\alpha-1}+\frac 1N D p^{\alpha-1-N}),
& 2\leq \alpha\leq N-1,
\\
\frac{N}{(N-1)m!} \res\limits_{p=0} p^{-1} m F^{m-1}(\frac{N-1}{N}A+\frac 1N B - \ch(m))
& \\ \qquad
+ \frac{N}{(N-1)m!} \res\limits_{p=0} F^{m}
(\frac{N-1}{N}C+\frac 1N D p^{-N}),
& \alpha=1,
\\
\frac{N}{(N-1)m!} \res\limits_{p=0} p^{-2} m F^{m-1}(\frac{N-1}{N}A+\frac 1N B - \ch(m))
& \\ \qquad
+ \frac{1}{m!} \res\limits_{p=0} F^{m}
(-Cp^{N-1}+D p^{-1}+\frac N{N-1} D p^{-N-1}),
& \alpha=N.
\end{cases}
\end{align}
For instance, for $m=0$ we have $(S_m)^\alpha_\beta = \delta^\alpha_\beta$.
\begin{remark}
In the $N=2$ case we get
\begin{equation}
S_1 =
\begin{pmatrix}
0 & 0 \\ 1& 0
\end{pmatrix}
\end{equation}
which means that the calibration is fixed to $\psi = 0$ in the notations of~\cite{carletHigherGeneraCatalan2021}.
\end{remark}

\section{Partition function and enumerative meaning}
\label{sec:PartitionFunction}

\subsection{Ancestor and descendant potentials} 

The purpose of this section is to briefly recall the Givental formulas for the partition function associated with the Dubrovin--Frobenius manifolds that we consider.  We assume the reader to be familiar with Givental's quantization and in particular with the standard definitions of all involved operators. We refer to the original papers~\cite{Giv-Semisimple,Giv2001} for a general exposition and to~\cite{carletHigherGeneraCatalan2021} for the conventions used in this paper. 

Let $\tau_{KdV}( \{ t_d\}_{d\geq 0},\epsilon^2)$ be the string KdV tau function, also referred in the literature as the Kontsevich--Witten tau function, 
and let $Q^i_d$, $i=1,\dots,N$, $d=0,1,2,\dots,$ be formal variables related to the standard descendant variables $t_d$ by the so-called dilaton shift $Q^i_d = t_d - \delta_{1,d}$. 
The total ancestor potential is given by
\begin{equation}
\label{eq-ancestor}
\cA(\{q^\alpha_d\}_{\substack{\alpha=1,\dots,N;\, d\geq 0}}) \coloneqq \hat \Psi \hat R \prod_{i=1}^N \tau_{KdV}(\{ Q^i_d+\delta_{1,d} \}_{d \geq 0},\epsilon^2). 
\end{equation}
Here $\hat R$ is the quadratic differential operator obtained by Givental quantization of the $R$ matrix defined in Section~\ref{sec:R-matrix}; we won't however need its explicit form in this paper.  
The operator $\hat \Psi$ is just the change of variables in the resulting function given by $Q^i_d = \Psi^i_\alpha q^\alpha_d$. 

The total descendant potential is defined as 
\begin{equation} \label{eq:D-in-q}
	\cD(\{t^\alpha_d\}_{\substack{\alpha=1,\dots,N;\, d\geq 0}}) \coloneqq C \hat{S}^{-1}\cA(\{q^\alpha_d\}_{\substack{\alpha=1,\dots,N;\, d\geq 0}}), 
\end{equation}
where $\hat{S}^{-1}$ is the operator obtained by quantization of the $S$-matrix coming from the calibration, and $C$ is a normalization function on the underlying Dubrovin--Frobenius manifold with the property $C(t_{sp})=1$. The variables $t^\alpha_d$ are related to  $q^\alpha_d$ by the dilaton shift $q^\alpha_d = t^\alpha_d - \delta^{\alpha,1}\delta_{d,1}$. 

	Throughout the paper, in order to simplify the exposition and the computations, we are only interested in $\cA$ and $\cD$ computed at the special point $t_{sp}$. For this reason we also omit the discussion of the dependence of $\cA$ and $\cD$ on the choice of point $t\in \hat M$, see~\cite{Giv2001} or e.~g.~\cite{carletHigherGeneraCatalan2021} for details. 

It is also useful to consider representation of Equation~\eqref{eq:D-in-q} entirely in the variables $t^\alpha_d$ and $T^i_d = Q^i_d + \delta_{1,d}$. We have:
\begin{align} \label{eq:D-in-t}
	\cD(\{t^\alpha_d\}_{\substack{\alpha=1,\dots,N;\, d\geq 0}}) = C \,{}^t\! \hat{S}^{-1} \hat \Psi  \,{}^t\! \hat R \prod_{i=1}^N \tau_{KdV}(\{ \Delta_i^{\frac 12}T^i_d \}_{d \geq 0}, \Delta_i \epsilon^2),
\end{align}
where $\Delta_i^{-\frac 12} = \Psi^i_1$ for $i=1, \dots , N$ and 
\begin{align}
	 \,{}^t\! \hat{S} & = \Big(e^{-\frac{\partial}{\partial q^1_1}} \hat{S} e^{\frac{\partial}{\partial q^1_1}}\Big)\Big|_{q^\alpha_d \to t^{\alpha}_d}, & \\
	 \,{}^t\! \hat R & = \Big(e^{-\sum_{i=1}^N \Psi^i_1 \frac{\partial}{\partial Q^i_1}} \hat{R} e^{\sum_{i=1}^N \Psi^i_1 \frac{\partial}{\partial Q^i_1}}\Big)\Big|_{Q^i_d \to T^{i}_d} .
\end{align}

We have hence two different representations of the total descendant potential expanded at the special point $t_{sp}$. The advantage of Equation~\eqref{eq:D-in-q} is that it is very convenient for the proof of Hirota equations in Section~\ref{sec:HirotaAncestor}. On the other hand, Equation~\eqref{eq:D-in-t} represents the total descendant potential via formal operations on formal power series, and in this form it is used to reveal its enumerative meaning in Section~\ref{sec:RootedHypermaps}.

\subsection{Relation to rooted hypermaps} 
\label{sec:RootedHypermaps} 

For a given $N\geq 2$, a $(k_1+1,\dots,k_n+1)$-hypermap of genus $g$ is a way to combinatorially glue an oriented genus $g$ surface from $n$ white polygons with $k_1+1,\dots,k_n+1$ sides, respectively, and $\sum_{i=1}^n (k_i+1)/N$ black $N$-gons (their number is assumed to be integer, of course). The polygons are glued by identifying their sides in pairs, one from a white polygon, and one from a black $N$-gon. 
A $(k_1+1,\dots,k_n+1)$-hypermap of genus $g$ is called \emph{rooted} if one side for each white $(k_i+1)$-gon is distinguished. 

Let $\mathsf{RHM}_{g;k_1+1,\dots,k_n+1}$ be the number of the isomorphism classes of  $(k_1+1,\dots,k_n+1)$-hypermaps of genus $g$. 
This number has various further interpretations in the literature, through a number of bijective identifications and dualities: up to a factor of $\prod_{i=1}^n (k_i+1)$ it is related to enumeration of $N$-orbifold strictly monotone Hurwitz numbers, special kind of Belyi functions, or constellations.

Our goal it to relate the descendant potential to enumeration of hypermaps, thus assigning it an enumerative meaning. The main result is the following:

\begin{theorem}\label{thm:Hypermaps} Consider the expansion of the descendant potential at the special point, with the calibration given by Equation~\eqref{eq:S-at-tsp}. We have:
\begin{equation}
\log\mathcal{D} \big|_{\substack{t^{\alpha}_d = \delta^{\alpha,N}\delta_{d,0}
\\ \alpha=2,\dots,N; d\geq 0}} = 
\sum_{g=0}^\infty \sum_{n=1}^\infty \frac{\epsilon^{2g-2}}{n!} \sum_{k_1,\dots,k_n=0}^\infty \mathsf{RHM}_{g;k_1+1,\dots,k_n+1} \prod_{i=1}^n \frac{t^{1}_{k_i}}{(k_i+1)!} .
\end{equation}
\end{theorem}

\begin{proof}
In order to prove this theorem, we recall that enumeration of rooted hypermaps can be resolved by a matrix model computation, and in particular can be computed by expanding the symmetric $n$-differentials obtained via the Chekhov--Eynard--Orantin topological recursion, see~\cite{Eynard-Book, DOPS18}.  

We consider the spectral curve given by the data $X=z/(1+z^{N})$, $\tilde y=z^N$, and $\tilde B=dz_1dz_2/(z_1-z_2)^2$, $\tilde \omega_{0,1}=\tilde ydX/X$. This choice of the spectral curve data is standard from the point of view of the KP integrability, see~\cite{ACEH20,BDKS20}.  It is more convenient for us to us an equivalent form given by  $x=z^{N-1}+1/z$, $y=-z$, $B=dz_1dz_2/(z_1-z_2)^2$, and $\omega_{0,1}=ydx = \tilde \omega_{0,1}-dx/x$; thus $x(z) = f(z,t_{sp})$. 
 
It is proved (see e.~g.~\cite{Eynard-Book,DOPS18,ACEH20,BDKS20}) that the CEO topological recursion applied to this spectral curve produces the symmetric differentials $\omega_{g,n}(z_1,\dots,z_n)$, $2g-2+n>0$, from initial data $\omega_{0,1}(z_1)$, $\omega_{0,2}(z_1, z_2)$,  whose expansions near $z_1=\cdots=z_n=0$ in the variables $x_i=x(z_i)$, $i=1,\dots,n$, are given by
\begin{align} \label{eq:Omega-enumerative}
	\omega_{g,n}= (-1)^n \sum_{k_1,\dots,k_n\geq 0} \mathsf{RHM}_{g;k_1+1,\dots,k_n+1} \prod_{ i=1}^n \frac{d x_i}{x_i^{k_i+2}}.
\end{align}
 
Recall that $c_1,\dots,c_N$ are the critical points (identified with their $z$-coordinates) of $x(z)$. Define
\begin{equation}
	\xi^j(z) \coloneqq  \frac{dz}{d \sqrt{2(x(z)-x(c_j))}}\Bigg|_{z=c_j} \cdot \frac{1}{c_j-z}, \qquad j=1,\dots,N,
\end{equation}
where the choice of the square roots $\sqrt{2(x(z)-x(c_j))}$ is aligned with the choices made for $\Delta_j$, $j=1,\dots,N$, namely, we require that
\begin{equation}
	\frac{dz}{d\sqrt{2(x(z)-x(c_j))}}\Bigg|_{z=c_j} = \Delta_j^{-1/2}, \qquad j=1,\dots,N.
\end{equation}
It is proved in~\cite{db19} that 
\begin{align} \label{eq:Omega-Givental}
	&
	\omega_{g,n} = 
	\\ \notag & [\epsilon^{2g-2}] 
	 \sum_{\substack{1\leq i_1,\dots,i_n \leq N\\ a_1,\dots,a_n\geq 0}} \Bigg(\Big(\prod_{j=1}^n \frac{\partial}{\partial T^{i_j}_{a_j}}\Big) \log \,{}^t\! \hat R \prod_{i=1}^N \tau_{KdV}(\Delta_i^{\frac 12}T^i_d, \Delta_i \epsilon^2)\Bigg)\Bigg|_{T^i_a = 0} \prod_{j=1}^n d\Big(-\frac{d}{dx_j}\Big)^{a_j} \xi^{i_j}(z_{j})
\end{align}
in the stable range. 
Combining Equations~\eqref{eq:Omega-enumerative} and~\eqref{eq:Omega-Givental} we prove the theorem. It requires some computation performed below, which is a straightforward generalization of similar computations done in~\cite{carletHigherGeneraCatalan2021,DOSS}.

First, we obtain an equivalent form of Equation~\eqref{eq:Omega-Givental} in the flat frame. To this end, define $\tilde{\xi}^\alpha \coloneqq (\Psi^{-1})^\alpha_i \xi^i$, $\alpha=1,\dots,N$. We have: 
\begin{align} \label{eq:Omega-Givental-Flat} 
	&
	\omega_{g,n} = 
	\\ \notag &  [\epsilon^{2g-2}]
	\sum_{\substack{1\leq \alpha_1,\dots,\alpha_n \leq N\\ a_1,\dots,a_n\geq 0}} \Bigg(\Big(\prod_{j=1}^n \frac{\partial}{\partial t^{\alpha_j}_{a_j}}\Big) \log \hat\Psi \,{}^t\! \hat R \prod_{i=1}^N \tau_{KdV}(\Delta_i^{\frac 12}T^i_d, \Delta_i \epsilon^2)\Bigg)\Bigg|_{t^\alpha_a = 0} \prod_{j=1}^n d\Big(-\frac{d}{dx_j}\Big)^{a_j} \tilde \xi^{\alpha_j}(z_{j})
\end{align}
(in this form Equation~\eqref{eq:Omega-Givental-Flat} could be also directly derived from~\cite{DNOPS-Hurwitz}, since we consider here an instance of a Hurwitz Frobenius manifold in the sense of Dubrovin). Note that 
\begin{align}
	\tilde \xi^\alpha =
	\begin{cases}
	 \frac{z^\alpha}{1-(N-1)z^{N}} \cdot (N-1), & \alpha = 2,\dots,N-1,	\\
	 \frac{z}{1-(N-1)z^{N}}, & \alpha = 1,	\\
	 \frac{1}{1-(N-1)z^{N}}, & \alpha = N.	\\
	\end{cases}
\end{align} 

Second, we use the fact that the action of $\,{}^t\! \hat{S}$ amounts to a linear triangular change of variables combined with the shift of the point of expansion (with the implied correction of the unstable terms), that is
\begin{align} \label{eq:S-action-explained}
& 
		\left( \prod_{j=1}^n \frac{\partial}{\partial t^{\alpha_j}_{a_j}}  \log{\,}^t\!\hat{S}^{-1}\hat{\Psi}{\,}^t\!\hat{R}\prod_{i=1}^2\tau_{KdV}(\{\Delta_i^{1/2}T^i_a\}_{a\geq 0},\Delta_i\epsilon^2)\right)\Big|_{t^\alpha_a=\delta^{\alpha,N}\delta_{a,0}} 
		\\  \notag 
&
		= \sum_{\substack{0\leq \ell_j \leq a_j,\\j=1,\dots,n}} (S_{\ell_j})^{\beta_j}_{\alpha_j} \left( \prod_{j=1}^n \frac{\partial}{\partial t^{\beta_j}_{a_j-\ell_j}}  \log\hat{\Psi}{\,}^t\!\hat{R}\prod_{i=1}^2\tau_{KdV}(\{\Delta_i^{1/2}T^i_a\}_{a\geq 0},\Delta_i\epsilon^2)\right)\Big|_{t^\alpha_a = 0} .
\end{align}
With~\eqref{eq:S-action-explained},~\eqref{eq:Omega-Givental-Flat}, and~\eqref{eq:Omega-enumerative} combined, the statement of the theorem reduces in the stable range $2g-2+n>0$ to the following lemma:
\begin{lemma} We have:
	\begin{equation}
		\res_{z=\infty} \frac{x^{k+1}}{(k+1)!} d \Big(-\frac{d}{dx}\Big)^a \tilde \xi^\alpha = 
		\begin{cases}
			0, & a> k\geq -1, \\
			(S_{k-a})^\alpha_1, & k\geq a \geq 0.
		\end{cases}
	\end{equation}
\end{lemma}
\begin{proof} Recall Equation~\eqref{eq:S-at-tsp} for $\beta=1$. Note that 
	\begin{align}
		\res_{z=0} \frac{x^{k+1}}{(k+1)!} d \Big(-\frac{d}{dx}\Big)^a \tilde \xi^\alpha
		= - \res_{z=0} \frac{x^{k-a}}{(k-a)!} dx \tilde \xi^\alpha
	\end{align}
if $k\geq a$ and $0$ otherwise. Then for $m=k-a$ we rewrite the latter expression as 
	\begin{align}
	- \res_{z=0} \frac{x^{m}}{m!} dx \tilde \xi^\alpha = \res_{z=0} (N-1)z^{\alpha-2} \frac{x^{m}}{m!} dz = \frac{-(N-1)}{m!} \res_{z=\infty} z^{\alpha-2} x^{m}dz
\end{align}
in the case $\alpha =2,\dots, N-1$ and if $\alpha=1$ or $N$, then we have
	\begin{align} 
	- \res_{z=0} \frac{x^{m}}{m!} dx \tilde \xi^\alpha = \res_{z=0}z^{p-2} \frac{x^{m}}{m!} dz = \frac{-1}{m!} \res_{z=\infty} z^{p-2} x^{m}dz,
\end{align}
with $p=1$ for $\alpha=1$ and $p=0$ for $\alpha=N$. For all $\alpha$ we obtain the formulas that coincide with the ones given in Equation~\eqref{eq:S-at-tsp} for $\beta=1$.
\end{proof}

Finally, we have to check the unstable cases $(g,n)=(0,1)$ and $(g,n)=(0,2)$. To this end, we use explicit formulas for the unstable cases in terms of the $S$-matrix from~\cite{Giv-Semisimple}:
\begin{align}
	[\epsilon^{-2} t^1_a] \log \cD \big|_{\substack{t^{\alpha}_d = \delta^{\alpha,N}\delta_{d,0}
			\\ \alpha=2,\dots,N; d\geq 0}} & = \eta_{\unity,\alpha} (S_{a+2})^\alpha_1, \\ \notag 
	[\epsilon^{-2} t^1_at^1_b] \log \cD\big|_{\substack{t^{\alpha}_d = \delta^{\alpha,N}\delta_{d,0}
			\\ \alpha=2,\dots,N; d\geq 0}} & = [z^a w^b] \frac{-\eta_{11} + \sum_{m,n=0}^\infty (S_m)^\mu_1 z^m  (S_n)^\nu_1 w^n \eta_{\mu\nu}}{z+w} ,
\end{align}
which allow us to conclude the proof of the theorem with the following two lemmata.

\begin{lemma} We have: 
\begin{equation} \label{eq:ENU-01}
	\frac{1}{(k+1)!}\mathsf{RHM}_{0;k+1} = \eta_{\unity,\alpha} (S_{k+2})^\alpha_1	.
\end{equation}	
\end{lemma}

\begin{proof} 
Recall that the unit vector field is $\partial/\partial t^{N-1}$. Hence, the right hand side of Equation~\eqref{eq:ENU-01} is equal to
\begin{align}
	\eta_{\unity,\alpha} (S_{k+2})^\alpha_1 & = \frac{1}{N-1}(S_{k+2})^2_1 
	= \frac{1}{N-1} \cdot \frac{-(N-1)}{(k+2)!} \res_{p=\infty} \left(p^{N-1}+\frac 1p\right)^{k+2} dp\\
	\notag &
	 = \frac{1}{(k+2)!} \res_{p=0} \left(p^{N-1}+\frac 1p\right)^{k+2} dp.
\end{align} 
On the other hand, Equation~\eqref{eq:Omega-enumerative} for $g=0$, $n=1$ gives $\mathsf{RHM}_{0;k+1} = [X^{k+1}] \tilde y$, where $\tilde y = z^N = z(z^{N-1}+1/z)-1$, thus
\begin{align}
	\frac{1}{(k+1)!}\mathsf{RHM}_{0;k+1} & = \frac{1}{(k+1)!}\res_{z=0} z^N \frac{d (z/(1+z^N))}{(z/(1+z^N))^{k+2}} \\ \notag
	& = \frac{-1}{(k+1)!}\res_{z=0} \left(z\left(z^{N-1}+\frac 1z \right) -1 \right) \left(z^{N-1}+\frac 1z \right)^k d  \left(z^{N-1}+\frac 1z \right) \\ \notag
	& = \frac{1}{(k+2)!}\res_{z=0}  \left(z^{N-1}+\frac 1z \right)^{k+2} dz,
\end{align}
which implies the statement of the lemma.
\end{proof}

\begin{lemma} 
We have:
	\begin{align} \label{eq:ENU-02}
		& \frac{1}{(k_1+1)!(k_2+1)!}\mathsf{RHM}_{0;k_1+1,k_2+1} \\ \notag & = [w_1^{k_1}w_2^{k_2}]
		\frac{-\eta_{1,1}+\sum_{n_1,n_2=0}^\infty w_1^{n_1}w_2^{n_2}  \eta_{\alpha_1,\alpha_2} (S_{n_1})^{\alpha_1}_1 (S_{n_2})^{\alpha_2}_1	}{w_1+w_2}.
	\end{align}	
\end{lemma}

\begin{proof} 
First we analyze the right hand side of Equation~\eqref{eq:ENU-02}. Recall the definition of $\eta$ given in Equation~\eqref{eq:Eta-Definition}, in particular, $\eta_{1,1}=0$. Recall also the formulas for $(S_n)^\alpha_1$ at the special point given in Equation~\eqref{eq:S-at-tsp}. In particular, a useful version of this formula in the case $\alpha=N$ is 
\begin{equation}
(S_n)^\alpha_1 = -\frac{1}{n!}\res_{p=\infty} p^{-2} f(p,t_{sp})^{n} dp	= \frac{N-1}{n!}\res_{p=0} p^{N-2} f(p,t_{sp})^{n} dp	.
\end{equation}	
Using this expression, and, furthermore, replacing the residues at $p=\infty$ with the residues at $p=0$, we have:
\begin{align} \label{eq:SS-inProof02}
&
\sum_{n_1,n_2=0}^\infty w_1^{n_1}w_2^{n_2}  \eta_{\alpha_1,\alpha_2} (S_{n_1})^{\alpha_1}_1 (S_{n_2})^{\alpha_2}_1 
\\ \notag & 
= \sum_{n_1,n_2=0}^\infty w_1^{n_1}w_2^{n_2} (N-1) \sum_{\alpha=1}^N \res_{p_1=0} \res_{p_2=0} p_1^{\alpha-2} p_2^{N-1-\alpha} \frac{f(p_1,t_{sp})^{n_1}}{n_1!}\frac{f(p_2,t_{sp})^{n_2}}{n_2!} dp_1dp_2.
\end{align}
Obviously, the constant term of this series is equal to zero. 
	
Now, to compute the left hand side of Equation~\eqref{eq:ENU-02}, we recall that from Equation~\eqref{eq:Omega-enumerative} for $g=0$, $n=2$ we have that
	\begin{align}
		& \frac{1}{(k_1+1)!(k_2+1)!}\mathsf{RHM}_{0;k_1+1,k_2+1} \\ \notag 
		& = \res_{z_1=0}\res_{z_2=0} \frac{x(z_1)^{k_1+1}}{(k_1+1)!}\frac{x(z_2)^{k_2+1}}{(k_2+1)!}\cdot 
		d_{z_1} d_{z_2} \Big(\log(z_1-z_2)-\log\big(x(z_1)-x(z_2)\big)\Big)
		\\ \notag 
		& = -\res_{z_1=0}\res_{z_2=0} \frac{x(z_1)^{k_1+1}}{(k_1+1)!}\frac{x(z_2)^{k_2+1}}{(k_2+1)!}\cdot d_{z_1} d_{z_2} \log\left(1-z_1z_2\frac{z_1^{N-1}-z_2^{N-1}}{z_1-z_2}\right).
	\end{align}
Hence, 
	\begin{align}
	&(w_1+w_2) \sum_{k_1,k_2\geq 0} w_1^{k_1} w_2^{k_2} \frac{1}{(k_1+1)!(k_2+1)!}\mathsf{RHM}_{0;k_1+1,k_2+1} \\ \notag 
	& = \sum_{\substack{k_1,k_2\geq 0\\ k_1+k_2\geq 1}} w_1^{k_1} w_2^{k_2} \res_{z_1=0}\res_{z_2=0} \frac{x(z_1)^{k_1}}{k_1!}\frac{x(z_2)^{k_2}}{k_2!}\cdot 
	\\ \notag & \quad 
	\left(
	d_{z_1} x(z_1) d_{z_2} \log\left(1-z_1z_2\frac{z_1^{N-1}-z_2^{N-1}}{z_1-z_2}\right)
	+ d_{z_2} x(z_2) d_{z_1} \log\left(1-z_1z_2\frac{z_1^{N-1}-z_2^{N-1}}{z_1-z_2}\right)
	\right)
	\\ \notag
		& = \sum_{\substack{k_1,k_2\geq 0\\ k_1+k_2\geq 1}} w_1^{k_1} w_2^{k_2} \res_{z_1=0}\res_{z_2=0} \frac{x(z_1)^{k_1}}{k_1!}\frac{x(z_2)^{k_2}}{k_2!}\cdot 
(N-1)\sum_{\alpha=1}^{N} z_1^{\alpha-2} z_2^{N-1-\alpha} dz_1dz_2.
\end{align}
The latter expression coincides with~\eqref{eq:SS-inProof02}, which implies the statement of the lemma. 
\end{proof}

These computations complete the proof of the theorem.	
\end{proof}

\subsection{Hypergeometric KP tau function} 

Consider the generating function $\cZ$ for the rooted hypermaps
\begin{equation}
	\cZ \coloneqq \exp \left( \sum_{g\geq 0} \epsilon^{2g-2} \sum_{n\geq 1} \frac{1}{n!} \sum_{k_1,\dots,k_n\geq 1} \mathsf{RHM}_{g,k_1,\dots,k_n} \prod_{i=1}^n \KPt_{k_i} \right).
\end{equation}
This function is a KP tau function of hypergeometric type~\cite{GouldenJackson} and it can be given by
\begin{equation}\label{eq:HyperGeomKPtau}
	\cZ = \left( \sum_\lambda s_\lambda(\{\KPp_i\}_{i\geq 1}) s_\lambda(\{\tilde\KPp_i\}_{i\geq 1}) \prod_{(i,j)\in\lambda} (1+\epsilon(i-j))\right)\Bigg|_{\substack{\KPp_i = i\KPt_i/\epsilon,\ i\geq 1 \\
			\tilde\KPp_i = \delta_{iN}/\epsilon,\ i\geq 1}}\,.
\end{equation}
Here the sum is taken over the Young diagrams $\lambda$, and $s_\lambda(\{\KPp_i\}_{i\geq 1})$ and  $s_\lambda(\{\tilde\KPp_i\}_{i\geq 1})$ are the Schur functions considered in the power sums variables $\KPp_1,\KPp_2,\dots$ and $\tilde\KPp_1,\tilde\KPp_2,\dots$, respectively. 

\begin{corollary} 
\label{cor:KP} 
The restriction of the total descendant potential 
\begin{equation}
	\cD \big|_{\substack{t^1_d = (d+1)! \KPt_{d+1},\ d\geq 0 \\ t^{\alpha,d} = \delta^{\alpha,N}\delta^{d,0},\ \alpha=2,\dots,N; d\geq 0}}	
\end{equation}
 is a hypergeometric KP tau-function given by~\eqref{eq:HyperGeomKPtau}.
\end{corollary}

\begin{remark} 
The function $\cZ$ is expanded in the power sum variables $\KPp_i$, $i\geq 1$, as 
	\begin{equation}
		\cZ =\exp \left( \sum_{g\geq 0} \epsilon^{2g-2} \sum_{n\geq 1} \frac{1}{n!} \sum_{k_1,\dots,k_n\geq 1} \mathsf{OSMH}_{g,k_1,\dots,k_n} \prod_{i=1}^n \KPp_i \right).
	\end{equation}
Here $\mathsf{OSMH}_{g,k_1,\dots,k_n}$ are the so-called $N$-orbifold strictly monotone Hurwitz numbers, see e.g.~\cite{HarnadOrlov,ALS}.
\end{remark}

\section{Period vectors and vertex operators} 
\label{sec:PeriodVectors} 

Following the presentation in~\cite[Section 5]{dub-Painleve}, it is convenient to think of the superpotential $\lambda = f(p,t)$ as a multi-valued function $p=p(\lambda,t)$ defined on the cut $\lambda$-plane. This allows to fix the choices needed to describe very explicitly the period vector. The analysis of this section is valid in a neighbourhood of the special point $t_{sp}$, with explicit computations performed only at $t_{sp}$.

\subsection{Cuts in the \texorpdfstring{$\lambda$}{lambda}-plane and one-point cycles}
\label{sec:cuts} 

Recall that the critical points of $f(p,t_{sp})$ are given by 
$c_i = (N-1)^{-1/N}\exp(2\pi\i \cdot i/N)$, $i=1,\dots,N$ and the values of the canonical coordinates are equal to 
$u^i = f(c_i,t_{sp}) = N(N-1)^{\frac 1N -1}\exp(-2\pi \i \frac iN)$, $i=1,\dots,N$.
The line $\ell_i:=e^{2\pi \i \frac{i}N}\R_+$ passing through the critical point $c_i$ covers twice the half line $\tilde{\ell}_i:=u^i + e^{-2\pi \i \frac{i}N}\R_+$, where $i=1,\dots , N$.
Let us denote by $D_{i-1}$ the sector with $2\pi \frac{i-1}N<\arg(p)<2\pi \frac{i}N$ cut out by $\ell_{i}$ and $\ell_{i-1}$. 
Clearly two preimages of a point $\lambda$ of  $\tilde{\ell}_i$ are contained on $\ell_i$ and tend to zero  and $\infty$ respectively when $|\lambda| \to \infty$; the remaining $N-2$ preimages are contained in the sectors $D_j$ with $j\not= i, i-1$. Each sector $D_i$ is mapped biholomorphically to the complex plane minus $\tilde{\ell}_i \cup \tilde{\ell}_{i+1}$.
Let $\gamma_i$ be the deck transformation induced by a small path going counterclockwise around the critical value $u^i$; then we  have
\begin{equation}
\gamma_i : D_{i-1} \leftrightarrow D_i, \quad i=1, \cdots , N
\end{equation}
while $\gamma_i$ leaves invariant  $D_j$ with $j\not=i, i-1$. 
We denote by $p_j(\lambda)$ the preimage of $\lambda$ in $D_j$, $j=0,\dots,N-1$, with the additional convention that $D_N=D_0$ and  $p_N(\lambda)=p_0(\lambda)$.

\subsection{Period vectors} 

We define the one-point period vectors $I^{(0)}_{p_i}(\lambda,t)$ as
\begin{equation}
\label{eq:period1cycle}
	\left(I^{(0)}_{p_i}(\lambda,t_{sp})\right)^\alpha \coloneqq \frac 12 \eta^{\alpha\beta} \int_{p_i} \frac{\partial f(p,t_{sp})}{\partial t^{\beta}}  \frac{dp}{d_pf (p,t_{sp})} = -\frac 12 \eta^{\alpha\beta} \frac{\partial p_i(\lambda,t_{sp})}{\partial t^\beta},
\end{equation}
for $i=0,\dots,N-1$.
These period vectors are multi-valued vector-valued functions that in general don't solve the  Fuchsian system. They are single-valued on the cut $\lambda$-plane and have the same monodromy as the one-point cycles described above. 

We define the period vectors 
$I^{(-1)}_{e_i}(\lambda,t_{sp})$, $i=1,\dots,N$ corresponding to vanishing cycles as
\begin{align}
	(I^{(-1)}_{e_i}(\lambda,t_{sp}))^\alpha =  \frac 12\eta^{\alpha\beta} \int_{p_{i-1}(\lambda)}^{p_i(\lambda)} \frac{\partial f}{\partial t^{\beta}} (p,t_{sp}) dp
\end{align}
where the integration path in the $p$-plane is contained in $D_{i-1}\cup D_i$. 

Notice that we can write
\begin{align}
(I^{(-1)}_{e_i}(\lambda,t_{sp}))^\alpha &= (I^{(-1)}_{p_i}(\lambda,t_{sp}))^\alpha - (I^{(-1)}_{p_{i-1}}(\lambda,t_{sp}))^\alpha \\
&+ \frac12 \eta^{\alpha\beta}  \int_{u_{i-1}}^{u_i} \frac{\partial f(p_{i-1}(\rho, t_{sp}),t_{sp})}{\partial t^{\beta}}  
\frac{1}{\partial_p f(p_{i-1}(\rho, t_{sp}),t_{sp}) } d\rho
\end{align}
where
\begin{equation}
 (I^{(-1)}_{p_i}(\lambda,t_{sp}))^\alpha =  \frac12 \eta^{\alpha\beta}  \int_{u_{i}}^{\lambda} 
  \frac{\partial f(p_{i-1}(\rho, t_{sp}),t_{sp})}{\partial t^{\beta}}  
\frac{1}{\partial_p f(p_{i-1}(\rho, t_{sp}),t_{sp}) } d\rho.
\end{equation}
One can easily check that $\partial_\lambda I^{(-1)}_{p_i}(\lambda,t_{sp}) = I^{(0)}_{p_i}(\lambda,t_{sp})$
and this implies that 
\begin{align}
	I^{(0)}_{e_i}(\lambda,t_{sp})  := \partial_\lambda I^{(-1)}_{e_i}(\lambda,t_{sp}) = I^{(0)}_{p_i}(\lambda,t_{sp})- I^{(0)}_{p_{i-1}}(\lambda,t_{sp}).
\end{align}

For $\ell\in\Z$ and $i=1,\dots,N$, define
\begin{equation}
	I^{(\ell)}_{e_i} (\lambda,t_{sp})  := (\partial_\lambda)^{\ell +1} I^{(-1)}_{e_i}(\lambda,t_{sp}),
\end{equation}
where, for $\ell<0$, the operator $\partial_\lambda^{-1}$ in this formula denotes integration along a path in the cut $\lambda$-plane from $u_i$ to $\lambda$.
It can be checked that the period vectors $I^{(\ell)}_{e_i}$  solve the (generalized) Fuchsian equation associated to the Frobenius manifold and have the correct asymptotic behaviour in the neighbourhood of the critical values $u_i$.

For future reference it is important to compute the period vector corresponding to $w=\sum_{i=1}^N e_i$, i.e.
\begin{align}
	(I^{(-1)}_w (\lambda,t_{sp}))^\alpha & = \frac 12\sum_{i=1}^N \int_{p_{i-1}(\lambda)}^{p_i(\lambda)} \eta^{\alpha\beta} \frac{\partial f}{\partial t^{\beta}} (p,t_{sp}) dp.
\end{align} 
Using the topology of the covering  described in Section~\ref{sec:cuts}, we see that this formula is equal to an integral in the $p$-plane along a contour homotopy equivalent to the positively oriented loop around $p=0$. Notice also that  
\begin{align}
	\eta^{\alpha\beta}\frac{\partial f}{\partial t^{\beta}} (p,t_{sp}) = 
	\begin{cases}
		(N-1)p^{\alpha-2} & \alpha = 2,\dots,N-1, \\
		p^{-1} & \alpha=1, \\
		p^{-2} & \alpha = N.	
	\end{cases}
\end{align}
Therefore, 
$(I^{(-1)}_w (\lambda,t_{sp}))^\alpha =\pi\i \delta^{\alpha}_{1}$.

\subsection{Monodromy period vectors and choice of orbit}

The monodromy of the period vectors $I^{(0)}_{e_i} =  I^{(0)}_{p_i} -  I^{(0)}_{p_{i-1}}$ is determined by the monodromy of the one-point cycle period vectors which coincides with the monodromy of the one-point cycles. The fundamental group of the pointed $\lambda$ plane acts via $\gamma_i I^{(0)}_{v} =  I^{(0)}_{\gamma_i v}$ where
\begin{equation}
\gamma_i e_j = e_j - 2 G_{ij} e_i
\end{equation}
and the intersection matrix $G$ is given by $G_{11}=\cdots=G_{NN}=1$, and $G_{ij} = -1/2$ for $i=j\pm1$ and for $i=1,j=N$ and $i=N,j=1$.



Our goal is to choose a reasonably small orbit of the action of the monodromy group. To this end, let $w=\sum_{i=1}^N e_i$ and consider the vectors $v_0,\dots,v_{N-1}\in \langle e_1,\dots,e_{N-1} \rangle$ given by
\begin{align}
	v_i \coloneqq \sum^{N-1}_{\substack{j=0}} (p_i - p_j) 
	= Np_i-\sum_{j=0}^{N-1} p_j = \sum_{j=1}^i je_j + \sum_{j=i+1}^{N-1} (j-N) e_j .
\end{align}
Notice that $\gamma_i w =w$,  $\gamma_i\colon  v_{i-1} \leftrightarrow v_i$ and $\gamma_i v_k = v_k$, $k\not=i-1,i$ for  $i=1,\dots, N$ . In particular, the set $\{v_0,\dots,v_{N-1}\}\subset \langle e_1,\dots,e_{N-1}\rangle$ is closed under the subgroup of the monodromy group generated by $\gamma_1,\dots,\gamma_{N-1}$.  Notice in particular that $\gamma_N v_i = v_i$, $i=1,\dots,N-2$, $\gamma_N v_{N-1} = v_0+Nw$ and $\gamma_N v_0 = v_{N-1}-Nw$. 

  The whole orbit under the monodromy group that we consider is given by $\{v_0, v_1, v_2,\dots,v_{N-1}\} + N\mathbb{Z}w$, and the nontrivial actions of $\gamma_i$, $i=1,\dots, N$, can be schematically represented as
\begin{align}\label{eq:NontrivMonodromy}
\xymatrix@W=50pt{
	& & \cdots  \ar@{<->}[r]^{\gamma_{N-1}}& v_{N-1}-2Nw \ar@{<->}@(d,u)  [dlll]_{\gamma_N} 
	\\
v_0-Nw \ar@{<->}[r]^{\gamma_1} & v_1-Nw \ar@{<->}[r]^{\gamma_2} & \cdots  \ar@{<->}[r]^{\gamma_{N-1}}& v_{N-1}-Nw \ar@{<->}@(d,u)  [dlll]_{\gamma_N} 
\\
v_0\ar@{<->}[r]^{\gamma_1} & v_1 \ar@{<->}[r]^{\gamma_2} & \cdots  \ar@{<->}[r]^{\gamma_{N-1}}& v_{N-1}  \ar@{<->}@(d,u)  [dlll]_{\gamma_N}
\\
v_0+Nw \ar@{<->}[r]^{\gamma_1} & v_1+Nw \ar@{<->}[r]^{\gamma_2} & \cdots  \ar@{<->}[r]^{\gamma_{N-1}}& v_{N-1}+Nw \ar@{<->}@(d,u)  [dlll]_{\gamma_N} 
\\
v_0+2Nw \ar@{<->}[r]^{\gamma_1} & \cdots & &
}
\end{align}

\subsection{The function \texorpdfstring{$\cW$}{curly W} }

Given $a,b \in \C^N$ we define the function
\begin{equation}
\cW_{a,b} (t,\lambda)  = (I_a^{(0)}(t,\lambda) , I_b^{(0)}(t,\lambda) ).
\end{equation}

The fundamental group of the pointed plane acts on the integral of $\cW_{a,a}$ as follows 
\begin{lemma} \label{lem:MonodromyPhaseGeneralStatement}
For $a\in \C^N$ we have that 
\begin{equation}
\gamma_i \left( \int_{\lambda_0}^{\lambda} \cW_{a,a}d\rho \right) = 
\int_{\lambda_0}^\lambda \cW_{\gamma_i a, \gamma_i a}  d\rho 
+ 4 <a,e_i> \int_{\lambda_0}^{u^i} \cW_{\pi_ia,e_i} d\rho+\pi \i <a,e_i>^2 .
\end{equation}
\end{lemma}
Here $\pi_i a=(a+\gamma_ia)/2$ is the projection onto the hyperplane invariant under $\gamma_i$.
This is a general formula valid for arbitrary Frobenius manifolds (with $I^{(0))}_a(t,\lambda)$ being the normalized solutions of Dubrovin's Fuchsian system, see~\cite[Section 5]{dub-Painleve}), once we have  $<e_i,e_i>=1$ for $i=1,\dots,N$. The proof is reduced to a local analysis of the action of the monodromy and paths of the integration.
\subsection{A choice of orbit and covariant coefficients} 
\label{sec:OrbitCoefficients}

Let $\overline O$ be the orbit of the action of the monodromy group given by $\overline O\coloneqq \{\frac 2N(v_i + \ell Nw) | i=0,\dots,N-1, \ell\in \mathbb{Z}\}$. Our goal is to assign to the vectors of the orbit $\overline{O}$ certain coefficients that are covariant with respect to the action of the monodromy group. The peculiar choice of the coefficient $2/N$ in front of $v_i$ will become clear later, in the proof of Lemma~\ref{lem:coefficients}, cf.~Equation~\eqref{eq:proofLemmacoeff-cW}.


We define the coefficients $c_{i,\ell}$ as
\begin{align}
c_{i,\ell}(\lambda) & := d_{i,\ell}\exp \left[-\int_{\lambda_0}^\lambda \cW_{\frac 2N v_i,\frac 2N v_i}d\lambda \right], & i=0,\dots,N-1,
\end{align}
with 
\begin{align}
	d_{i,\ell} \coloneqq \exp \left[-\int_{p_0(\lambda_0)}^{p_i(\lambda_0)} \left(I_{2p-\frac 2N\sum_{j=1}^N p_j(\lambda(p))}^{(0)}(t,\lambda) , I_{2p-\frac 2N\sum_{j=1}^N p_j(\lambda(p))}^{(0)}(t,\lambda)  \right) d\lambda(p)\right],
\end{align}
where $I_{2p-\frac 2N\sum_{j=1}^N p_j(\lambda(p))}^{(0)}(t,\lambda)$ is the corresponding linear combination of one-point cycle period vectors defined as in \eqref{eq:period1cycle}.

We have to specify the contour in this definition. Consider a disk in the $p$-plane around the finite pole $0$ of $f(p,t_{sp})=p^{N-1}+p^{-1}$
 and all its critical points, with the points $p_0(\lambda_0),\dots,p_{N-1}(\lambda_0)$ on its boundary (the points are ordered
 according to the counterclockwise direction of its boundary). The contour connecting $p_0(\lambda_0)$ and $p_i(\lambda_0)$ in the definition of $d_{i,\ell}$ first rotates $\ell$ times along the boundary of this disk ($\ell$ times in the positive direction for $\ell\geq 0$, and $-\ell$ times in the negative direction for $\ell<0$), and then goes from $p_0(\lambda_0)$ to $p_i(\lambda_0)$ along the boundary in the positive direction. 


\begin{lemma} 
\label{lem:coefficients} 
{\ }
\begin{itemize}
\item[(1)] The function $\frac 2N(v_i+\ell Nw) \mapsto c_{i,\ell}$ is covariant on the orbit $\overline{O}$, that is, we have $\gamma_{i}\colon c_{i-1,\ell} \leftrightarrow c_{i,\ell}$, $i=1,\dots,N-1$, and $\gamma_N\colon c_{N-1,\ell}\leftrightarrow c_{0,\ell+1}$, and all other actions of $\gamma_i$, $i=1,\dots,N$ on $c_{j,\ell}$, $j=0,\dots,N-1$, are trivial. 	
\item[(2)] We have $c_{i,\ell} = c_{i} \exp(-\frac 2N 2\pi\i\ell)$, where $c_i\coloneqq c_{i,0}$, $i=0,\dots,N-1$.
\end{itemize}
\end{lemma}

\begin{proof} 
The fact that $\gamma_{i+1} c_{i,\ell} = c_{i+1,\ell}$, $i=0,\dots,N-2$, and $\gamma_{N} c_{N-1,\ell} = c_{0,\ell+1}$ follows directly from the definition.
	

Since $\lambda = p^{N-1}+ a_2p^{N-3} + \cdots +a_{N-1} + a_N/(p-a_1)$, we have $\sum_{i=0}^{N-1} p_i(\lambda, t) = a_1 = t^1$, and, therefore,
\begin{align}
\sum_{i=0}^{N-1}	\left(I^{(0)}_{p}(\lambda,t)\right)^\alpha\Big|_{p=p_i(\lambda)} = -\frac 12 \delta^{\alpha}_N.
\end{align}
Thus, 
\begin{align}\label{eq:proofLemmacoeff-cW}
	 \left(I_{2p-\frac 2N\sum_{j=1}^N p_j(\lambda(p))}^{(0)}(t,\lambda) , I_{2p-\frac 2N\sum_{j=1}^N p_j(\lambda(p))}^{(0)}(t,\lambda)  \right) d\lambda & = 
	\frac{\partial p}{\partial t^\alpha} \eta^{\alpha\beta} \frac{\partial p}{\partial t^\beta} d\lambda -\frac{2}N \frac{\partial p}{\partial t^N}d\lambda
	\\ \notag 
	& = d \log\Big(\frac{d\lambda}{dp}\Big) + \frac2N \frac{dp}{p-t^1}
\end{align}
(the second equality here is an explicit computation at $t_{sp}$).

Since $\oint d \log(d f/dp) \in 2\pi\i\Z$ for any closed contour in the $p$-plane and $\oint dp/(p-t^1)=2\pi\i$ for any closed contour that bounds a disk containing the finite pole of $f$ in the $p$-plane, we see that indeed $c_{i,\ell} = c_{i,0} \exp(-\frac 2N 2\pi\i\ell)$, $i=0,\dots,N-1$.


The assertion that $\gamma_{i+1}c_{i+1,\ell} = c_{i,\ell}$, $i=0,\dots,N-1$, and $\gamma_{N} c_{0,\ell} = c_{N-1,\ell-1}$, follows from the fact that the residues of~\eqref{eq:proofLemmacoeff-cW} at the critical points of $\lambda(p,t)$ in the $p$-plane are integers. Finally, all other actions of $\gamma_i$ on $ c_{j,\ell}$ are trivial since they are absorbed by the deformations of the integration contours that do not cross the poles of~\eqref{eq:proofLemmacoeff-cW}.
\end{proof}

\subsection{The vertex operators} 
\label{sec:VertexOperators}

For a vector $a \in \C^N=\langle e_1,\dots,e_N \rangle$ define
\begin{equation}
	\cf_a(t,\lambda,z) \coloneqq \sum_{l\in\Z}  I^{(l)}_a(t,\lambda) (-z)^l .
\end{equation}
Let $\Gamma^a$ denote the vertex operator defined as the exponential of the quantization of linear Hamiltonian $h_{\cf_a}$ of $\cf_a$:
\begin{equation}
\label{eq:defquant}
\Gamma^a \coloneqq e^{\widehat{\cf_a}}=e^{\widehat{(\cf_a)_-}}e^{\widehat{(\cf_a)_+}}.	
\end{equation}

\subsection{Asymptotics at \texorpdfstring{$\lambda \sim \infty$}{lambda sim infty}} 

Let $\Log\lambda$ be the main branch of the logarithm defined for the cut $\R_+$ in the $\lambda$-plane. Define $I_{e_1+\cdots+e_i,\infty}^{(0)}$, $i=1,\dots,N-1$ as
\begin{align}
	(I_{e_1+\cdots+e_i,\infty}^{(0)})^\alpha & = \frac 12\lambda^{(\alpha-N)/(N-1)}, & \alpha=2,\dots,N-1,
	\\ \notag
	(I_{e_1+\cdots+e_i,\infty}^{(0)})^1& =\frac{N}{2(N-1)}\lambda^{-1}, &
	\\ \notag
	(I_{e_1+\cdots+e_i,\infty}^{(0)})^N& =\frac 12, &
\end{align}
and $I^{(0)}_{w,\infty}=0$, where we recall that $w=\sum_{i=1}^Ne_i$.
Here the index $i$ of $e_i$ directs the choice of the branch of the logarithm:
\begin{align} \label{eq:ChoiceOfBranchOfLog}
	\log \lambda & = \Log \lambda + 2\pi\I (i-1), & \lambda^{1/(N-1)} = \exp\Big(\frac 1{N-1}\log \lambda\Big).
\end{align}

Let us now define, for $\ell\in\Z$ and $i=1,\dots,N$, 
\[
I^{(\ell)}_{e_i,\infty} := (\partial_\lambda)^{\ell } I^{(0)}_{e_i,\infty},
\]
where, for $\ell<0$, we use the following convention for the formal integration without constants:
\begin{equation}
	(\partial_\lambda^{-m})\lambda^{-1} = \frac 1{m!} \lambda^m (\log \lambda - \ch(m)),\qquad m\geq 0,
\end{equation}
and $(I_{w,\infty}^{(-1)})^\alpha = \pi\I\delta^{\alpha}_{1}$, so that we have
\begin{equation*}
(I_{w,\infty}^{(\ell)})^\alpha = \pi\I\delta^{\alpha}_{1} \frac{\lambda^\ell}{\ell!}, \quad \ell \geq 0. 
\end{equation*}
We also define
\[
\cf_{e_i,\infty}(\lambda,z) := \sum_{\ell\in\Z}  I^{(\ell)}_{e_i,\infty} (-z)^{\ell}.
\]
Let $\Gamma^v_\infty$, for $v\in \C^N$, denote the vertex operator defined as  quantisation of $\cf_{a,\infty}$ as above:
\begin{align}
	\Gamma^v_\infty \coloneqq e^{\widehat{\cf_{v,\infty}}}=e^{\widehat{(\cf_{v,\infty})_-}}e^{\widehat{(\cf_{v,\infty})_+}}.	
\end{align}
In particular, note that $	\Gamma^w_\infty = e^{\widehat{(\cf_{w,\infty})_-}}$.

\begin{lemma} 
\label{lem:Asymptotic-S}
	For $v\in\C^N$, the asymptotic expansion of $\cf_{v}$  for $|\lambda| \sim \infty$, $\arg \lambda \not=0$  is given by
	\begin{equation} \label{fSf}
		\cf_{v} (t, \lambda, z) \sim S(t, z)\ \cf_{v,\infty} (\lambda,z).
	\end{equation}
\end{lemma}
\begin{proof}
It is sufficient to prove that 
	\begin{equation} \label{expell}
		I^{(0)}_{e_1+\cdots+e_i}(t, \lambda) \sim \sum_{k=0}^\infty (-1)^k S_k I^{(k)}_{e_1+\cdots+e_i,\infty }, \quad i=1,\dots,N-1,
	\end{equation}
and it is sufficient to do this in a neighborhood of the special point $t_{sp}$. 
In this formula, the right hand side is a function
on the $\lambda$-plane with a branch cut along $\R_+$. For the left hand side we deform the cuts $\tilde{\ell}_i,~i=0,\ldots,N-1$ in such a way that they all become asymptotic in the upper half plane to $\R_+$ as $\lambda\to\infty$ in decreasing order $\tilde{\ell}_1,\ldots,\tilde{\ell}_{N}$ in ${\rm Im}(\lambda)$.

Equation~\eqref{expell} is equivalent to 
\begin{align}
	-\frac 12 \eta^{\alpha\beta} \frac{\partial}{\partial t^\beta} \left(p_i(\lambda,t)-p_0(\lambda,t)\right)
	\sim \sum_{k=0}^\infty (-1)^k (S_k)^\alpha_\beta \partial_\lambda^k (I^{(0)}_{e_1+\cdots+e_i,\infty })^\beta,
\end{align}
where by $p_j$, $j=0,\dots,N-1$, we mean the point that is in the preimage of $\lambda=f(p,t)$ in the disk $D_j$ (here we use that $t$ is close to $t_{sp}$ and we use the deformation of the cuts chosen in Section~\ref{sec:cuts}).

Recall the definition of $S$ given in Equation~\eqref{eq:S-matrix-general-def}, Section~\ref{sec:Calibration}. With this definition, it is sufficient to prove that 
 \begin{align}
  	-\frac 12 \left(p_i(\lambda,t)-p_0(\lambda,t)\right)
  	\sim \sum_{k=0}^\infty (-1)^k \theta_{\beta,k} \partial_\lambda^k (I^{(0)}_{e_1+\cdots+e_i,\infty })^\beta,
 \end{align}
Note that $p_i\to \infty$ for $\lambda\to \infty$, and in the corresponding sector we have (take into account the choice of the branch of $\log\lambda$ given in~\eqref{eq:ChoiceOfBranchOfLog})
\begin{align}
	{\lambda^{-\frac{1}{N-1}}} = \frac{p_i^{-1}}{\left(f(p_i,t)/p_i^{N-1}\right)^{\frac{1}{N-1}}}.
\end{align}
The Lagrange-B\"urmann formula implies that
\begin{align}
	p_i & = \sum_{n=1}^\infty \lambda^{-\frac{n}{N-1}} \res_{p=\infty}  \frac 1n f(p,t)^{\frac{n}{N-1}}dp \\ \notag
	& = \sum_{\beta=2}^{N-1} \sum_{k=0}^\infty \frac{(-\partial_\lambda)^k \lambda^{\frac{\beta-N}{N-1}} }{\prod_{i=0}^{k-1}(i-\frac{\beta-N}{N-1})} \frac 1n \res_{p=\infty} \frac{ f(p,t)^{k-\frac{\beta-N}{N-1}}}{(N-1)(k-\frac{\beta-N}{N-1})}dp
	\\ \notag 
	& \quad + \sum_{k=0}^\infty \frac{(-\partial_\lambda)^k \lambda^{-1} }{(k-1)!}  \res_{p=\infty} \frac{ f(p,t)^{k}}{(N-1)k}dp
	\\ \notag
	& = -2\sum_{\beta=2}^{N-1} \sum_{k=0}^\infty (-1)^k \theta_{\beta,k} \partial_\lambda^k (I^{(0)}_{e_1+\cdots+e_i,\infty })^\beta - \frac 2N \sum_{k=0}^\infty (-1)^k \theta_{1,k} \partial_\lambda^k (I^{(0)}_{e_1+\cdots+e_i,\infty })^1
\end{align}
(and one can formally add $0\cdot \sum_{k=0}^\infty (-1)^k \theta_{N,k} \partial_\lambda^k (I^{(0)}_{e_1+\cdots+e_i,\infty })^N$ to the latter expression). On the other hand, with $p_0\to t^1$ for $\lambda\to \infty$, we have:
\begin{align}
	\lambda^{-1} = \frac{p_0-t^1}{(p_0-t^1)f(p_0,t)}.
\end{align}
In this case, the  Lagrange-B\"urmann formula implies 
\begin{align}
p_0 -t^1 & = \sum_{n=1}^\infty \lambda^{-n} \res_{p=t^1} \frac 1n ((p_0-t^1)f(p_0,t))^n \frac{d(p-t^1)}{(p-t^1)^n} \\
& = - \sum_{k=0}^\infty \frac{(-\partial_\lambda)^k \lambda^{-1}}{(k-1)!} \res_{p=\infty} \frac 1k f(p_0,t)^k dp  \\
& = \frac{2(N-1)}{N} \sum_{k=0}^\infty (-1)^k \theta_{1,k} \partial_\lambda^k (I^{(0)}_{e_1+\cdots+e_i,\infty })^1.
\end{align}
Note also that $\theta_{N,0} = t^1$, hence $t^1 = 2 \theta_{N,0} (I^{(0)}_{e_1+\cdots+e_i,\infty })^N = 2 \sum_{k=0}^\infty (-1)^k \theta_{N,k} \partial_\lambda^k (I^{(0)}_{e_1+\cdots+e_i,\infty })^N$ for any $i=1,\dots,N-1$. 

Combining these computations, we obtain
\begin{align}
	-\frac 12 (p_i-p_0) = \sum_{k=0}^\infty (-1)^k \theta_{\beta,k} \partial_\lambda^k (I^{(0)}_{e_1+\cdots+e_i,\infty })^\beta, \quad i=1,\dots,N-1
\end{align}
in the corresponding sector of expansion.
\end{proof}

\subsection{Monodromy at \texorpdfstring{$\infty$}{infty}} 
\label{sec:asymptotics-and-monodromy-infty}

For future reference, it is useful to collect together the explicit formulas for $\cf_{ \frac 2N v_i,\infty}$ and the action of the monodromy around $\lambda=\infty$ on them. For $i=0$ we have:
\begin{align}\label{eq:finfty-0}
	(\cf_{ \frac 2N v_0,\infty})^1 & = \sum_{\ell\in\Z} (-z)^\ell \partial_\lambda^\ell \Big(-\partial_\lambda\Log \lambda\Big),
	\\ \notag  
	(\cf_{ \frac 2N v_0,\infty})^\alpha & = 0, & \alpha = 2,\dots,N-1,
	\\ \notag 
	(\cf_{ \frac 2N v_0,\infty})^N & =  \sum_{\ell\in\Z} (-z)^\ell \partial_\lambda^\ell \Big(-\frac 12\Big),
\end{align}
and for $i=1,\dots,N-1$ we have:
\begin{align}\label{eq:finfty-i}
	(\cf_{ \frac 2N v_i,\infty})^1 & = \sum_{\ell\in\Z} (-z)^\ell \partial_\lambda^\ell \Big(\frac{1}{N-1}\partial_\lambda\big(\Log \lambda+2\pi\i(i-1)\big)\Big),
	\\ \notag  
	(\cf_{ \frac 2N v_i,\infty})^\alpha & = \sum_{\ell\in\Z} (-z)^\ell \partial_\lambda^\ell \Big(\lambda^{\frac{\alpha-N}{N-1}} \exp\Big({2\pi\i(i-1)\frac{\alpha-N}{N-1}}\Big)\Big), & \alpha = 2,\dots,N-1,
	\\ \notag 
	(\cf_{ \frac 2N v_i,\infty})^N & =  \sum_{\ell\in\Z} (-z)^\ell \partial_\lambda^\ell \Big(\frac 1N\Big).
\end{align}
Let $\gamma_\infty=\gamma_1\cdots\gamma_N$ denote the monodromy along a contour around $\lambda=\infty$ oriented in the counterclockwise direction. The action of $\gamma_\infty$ on $\cf_{ \frac 2N v_i,\infty}$ is given by
\begin{align}
	(\gamma_\infty\cf_{ \frac 2N v_0,\infty})^1 & = (\cf_{ \frac 2N v_0,\infty})^1 -2\pi\i \sum_{\ell\in\Z} (-z)^\ell \partial_\lambda^{\ell+1} 1,
	\\ \notag  
	(\gamma_\infty\cf_{ \frac 2N v_0,\infty})^\alpha & = (\cf_{ \frac 2N v_0,\infty})^\alpha, & \alpha = 2,\dots,N,
\end{align}
or, in other words,
\begin{align}
	\gamma_\infty\cf_{ \frac 2N v_0,\infty} = \cf_{ \frac 2N v_0,\infty}  - 2 \cf_{w,\infty}.
\end{align}
For $i=1,\dots,N-2$ we have:
\begin{align}
	\gamma_\infty \cf_{ \frac 2N v_i,\infty} & = \cf_{ \frac 2N v_{i+1},\infty},
\end{align}
and finally 
\begin{align}
	(\gamma_\infty\cf_{ \frac 2N v_{N-1},\infty})^1 & = (\cf_{ \frac 2N v_1,\infty})^1 + 2\pi\i \sum_{\ell\in\Z} (-z)^\ell \partial_\lambda^{\ell+1} 1,
	\\ \notag  
	(\gamma_\infty\cf_{ \frac 2N v_{N-1},\infty})^\alpha & = (\cf_{ \frac 2N v_1,\infty})^\alpha, & \alpha = 2,\dots,N,
\end{align}
or, in other words,
\begin{align}
	\gamma_\infty\cf_{ \frac 2N v_{N-1},\infty} = \cf_{ \frac 2N v_1,\infty}  + 2 \cf_{w,\infty}.
\end{align}

\subsection{Conjugation by \texorpdfstring{$S$}{S}} 
\label{sec:conjugacybyS}

Let us now consider the conjugation of the vertex operator by the $S$ action of the Givental group.  
We define $\cW_{a,b}^{\infty} = ( I_{a,\infty}^{(0)}(t,\lambda) , I_{b,\infty}^{(0)}(t,\lambda) )$, $a,b\in \C^N$. Notice that 
\begin{align}
	\cW_{\frac 2N v_0,\frac 2N v_0}^{\infty} &=  \Big(2-\frac 2N\Big) \frac 1\lambda, \\
	\cW_{\frac 2N v_i,\frac 2N v_i}^{\infty} &=  \Big(1+\frac 1{N-1}-\frac 2N\Big) \frac 1\lambda, \qquad i =1,\dots,N-1.
\end{align}

For our choice of calibration, we have the following result:
\begin{proposition} 
\label{prop:conjugationS}
For $a\in\C^N$, and $t=t_{sp}$, we have 
\begin{equation}
		\hS \Gamma_\infty^{a} \hS^{-1} = e^{\frac12 \int_\lambda^\infty  \left( \cW_{a,a} - \cW_{a,a}^\infty \right) d\rho} \Gamma^a .
	\end{equation}
\end{proposition}
\begin{proof}
It follows from the Baker-Campbell-Hausdorff formula that, for $\cf = \sum_l \cf_l z^l$ in the loop space $\cV$ and $S(z)$ in the twisted loop group, we have
	\begin{equation}\label{eq:ConjugationByS}
		\hat{S} \widehat{e^{\cf}} \hat{S}^{-1} = e^{\frac12 W(\cf_+,\cf_+)} \widehat{e^{S \cf}},
	\end{equation}
	where 
	\begin{equation}
		W(\cf_+,\cf_+)= \sum_{k,l\geq 0} (W_{k,l} \cf_l,\cf_k ),
	\end{equation}
	and the coefficients $W_{k,l}\in \End(V)$ are defined by the generating formula
	\begin{equation} \label{defW}
		\sum_{k,l\geq 0} W_{k,l} w^{-k} z^{-l} = \frac{S^*(w) S(z) -1}{w^{-1} + z^{-1} }.
	\end{equation}
	
	We need therefore to evaluate the phase factor $W((\cf_{a,\infty})_+, (\cf_{a,\infty})_+)$. It follows from the symplectic properties of $S$ and Lemma~\ref{lem:Asymptotic-S} that for any $a,b \in \C^N$
	\begin{equation} \label{derlam}
		\frac{d}{d\lambda}W((\cf_{a,\infty})_+, (\cf_{b,\infty})_+)  = - \cW_{a,b} + \cW_{a,b}^\infty.
	\end{equation}
The right-hand side is a formal power series in powers of $\lambda^{-\frac1{N-1}}$ with leading order $\lambda^{-1 -\frac1{N-1}}$, so it can be integrated to a formal power series of the same type. The integration constant is  obtained by computing the $\lambda \to \infty$ limit in 
\begin{equation}
W((\cf_{a,\infty})_+, (\cf_{b,\infty})_+)  = \sum_{k,l\geq 0} (-1)^{k+l} (W_{kl} I_{a,\infty}^{(k)}, I_{b,\infty}^{(l)} ),
\end{equation}
which is equal to 
\begin{equation}
\lim_{\lambda \to \infty} (W_{00} I_{a,\infty}^{(0)}, I_{b,\infty}^{(0)} ) = (S_1)^1_N (I_a^{(0)})^N (I_b^{(0)})^N.
\end{equation}
From the properties of $S$ and the asymptotics of the periods we get 
\begin{equation}
(S_1)^1_N = 0 , \qquad 
(I_a^{(0)})^N =\frac12 (a_1 - a_{N}).
\end{equation}
\end{proof}

\section{The Hirota quadratic equations} 
\label{sec:HirotaAncestor}

In this section we  define the ancestor Hirota quadratic equations and prove that the ancestor potential $\cA$ satisfies them.  
Then, we prove the descendent Hirota quadratic equations for the descendant  potential $\cD$. 
In both proofs we assume the expansions at the special point $t_{sp}$.

\subsection{Definition of Hirota quadratic equations for the ancestor potential}  
\label{sec:def-hirota-ancestor}

Recall that the total ancestor potential $\cA$, defined in \eqref{eq-ancestor}, is a formal power series in the variables $q^\alpha_\ell+\delta^{a}_{N-1}\delta^1_\ell$ for $\alpha=1,\dots,N$ and $\ell \geq0$ whose coefficients are Laurent series in $\epsilon$, and which depends 
analytically on the point of $M$. 

In the infinite orbit $\overline O$ of the monodromy group defined in the beginning of Section \ref{sec:OrbitCoefficients} we choose 
a finite subset $O\coloneqq \{\frac 2N v_i  | i=0,\dots,N-1\}$, and we associate the functions $c_{i}=c_{i,0}$ defined in Section~\ref{sec:OrbitCoefficients} to the vectors in $O$. 

Recall that $(I^{(-1)}_w)^\alpha=\pi\i \delta^{\alpha}_{1}$, hence $(I^{(-1)}_w)_\alpha=\pi\i \delta_{\alpha,N}$ (where $w=\sum_{i=1}^N e_i$ and we use the scalar product to lower the index). Define 
\begin{equation}
	\cN = \exp\left( - \sum_{(j,\ell)\not= (N,0)}  \frac{(I^{(-\ell-1)}_{w})_j q^j_\ell }{(I^{(-1)}_{w})_N}  \frac{\partial }{\partial q^N_0}  \right)=\exp\left( -\sum_{\substack{\ell\geq 1\\ j=1,\dots,N}}  \frac{(I^{(-\ell-1)}_{w})_j q^j_\ell }{\pi\I}  \frac{\partial }{\partial q^N_0}  \right).
\end{equation}

\begin{lemma} \label{lem:SingleValued}
	If  $ q^N_0-\bar{q}^N_0 \in \epsilon (\mathbb{Z}+\frac 2N)$,
	the following expression is a single-valued $1$-form in $\lambda$:
	\begin{equation} \label{A-HQE}
		(\cN \otimes \cN)
		\left( \sum_{i=0}^{N-1} c_{i}\Gamma^{\frac 2N v_i} \otimes \Gamma^{-\frac 2Nv_i}\right) \left( \cA \otimes \cA \right) d\lambda .
	\end{equation}
	Here the two copies of $\cN$, $\Gamma$, and $\cA$ depend on the variables $q^i_\ell$ and $\bar{q}^i_\ell$ respectively.
\end{lemma}

\begin{proof}
	We need to prove that \eqref{A-HQE} is invariant under the action of the generators $\gamma_1, \dots,\gamma_N$ of the fundamental group of the pointed complex plane.  Note that the coefficients of $\cN$ are single-valued functions in $\lambda$. Indeed, since $I^{(-1)}_{w} $ is constant, all $I^{(-\ell-1)}_{w}$, $\ell\geq 0$, are polynomials in $\lambda$. 
	
Note that for $i=1,\dots,N-1$ we have $\gamma_i\colon O\to O$, $\gamma_i \colon c_{i-1} \leftrightarrow c_{i}$ and $\gamma_i \colon \Gamma^{\frac 2N v_{i-1}} \leftrightarrow \Gamma^{\frac 2N v_{i}}$, and the action of $\gamma_i$ on all other coefficients and vertex operators is trivial. Hence the action of $\gamma_i$, $i=1,\dots,N-1$, is just a reshuffling of the summands in~\eqref{A-HQE} which leaves~\eqref{A-HQE} invariant.

	For $\gamma_N$ we have $\gamma_N\colon c_{0}\mapsto c_{N-1}e^{\frac 2N 2\pi\i}$ and $\gamma_N\colon c_{{N-1}}\mapsto c_{0}e^{-\frac 2N 2\pi\i}$ and $\gamma_N c_{i} \mapsto c_{i}$ for $i=1,\dots,N-2$, by Lemma \ref{lem:coefficients}. In the meanwhile, $\gamma_N \Gamma^{\frac 2N v_i} = \Gamma^{\frac 2N v_i}$ for $i=1,\dots,N-2$, and
	\begin{align}\label{eq:BCHfw}
		& \gamma_N \Gamma^{\frac 2N v_0} =  \Gamma^{\frac 2N v_{N-1}-2w} = e^{-2\hat \cf_w} \Gamma^{\frac 2N v_{N-1}}, \\ \notag
		& \gamma_N \Gamma^{\frac 2N v_{N-1}} =  \Gamma^{\frac 2N v_0+2w} = e^{2\hat \cf_w} \Gamma^{\frac 2N v_0},
	\end{align}
	where we use the definition \eqref{eq:defquant} of the vertex operators $\Gamma$ and  the fact that $(\cf_w)_+=0$ for the second equalities. 
	Therefore, the action of $\gamma_N$ on~\eqref{A-HQE} is given by
	\begin{align} \label{eq:firststepgamma2AncestorHQE}
		&\gamma_N\,	(\cN \otimes \cN)
		\left( \sum_{i=0}^{N-1}  c_{i} \Gamma^{\frac 2N v_i} \otimes \Gamma^{-\frac 2N v_i}\right) \left( \cA \otimes \cA \right) d\lambda  \\  \notag 
		& =(\cN \otimes \cN)\, \Bigg( e^{-\frac 2N 2\pi\i} (e^{2\hat \cf_w}\otimes e^{-2\hat \cf_w}) c_{0} \Gamma^{\frac 2N v_0} \otimes \Gamma^{-\frac 2N v_0} 
		\\ \notag & \qquad \qquad  
		+e^{\frac 2N 2\pi\i} (e^{-2\hat \cf_w}\otimes e^{2\hat \cf_w}) c_{{N-1}} \Gamma^{\frac 2N v_{N-1}} \otimes \Gamma^{-\frac 2N v_{N-1}} 
		\\ \notag & \qquad \qquad 
		 +\sum_{i=1}^{N-2} c_{i} \Gamma^{\frac 2N v_i} \otimes \Gamma^{-\frac 2N v_i}\Bigg) \left( \cA \otimes \cA \right) d\lambda.
	\end{align}
	Now note that $\hat{f}_w = \epsilon^{-1}\sum_{\ell\geq 0} (I^{(-\ell-1)}_w)_\alpha q^\alpha_\ell$, and therefore $\cN e^{\pm 2\hat{\cf}_w} = e^{\pm 2\pi\i \epsilon^{-1} q^N_0} \cN$. Therefore, the right hand side of~\eqref{eq:firststepgamma2AncestorHQE} is equal to
	\begin{align}
		& \Bigg( e^{\frac{2\pi\i}{ \epsilon} (q^N_0-\bar{q}^N_0 - \frac{2\epsilon }{N})} (\cN \otimes \cN)  c_{0} \Gamma^{\frac 2N v_0} \otimes \Gamma^{-\frac 2N v_0}   \\ \notag 
		&+  e^{-\frac{2\pi\i}{ \epsilon} (q^N_0-\bar{q}^N_0 - \frac{2\epsilon }{N})} ( \cN \otimes \cN)   c_{{N-1}} \Gamma^{\frac 2N v_{N-1}} \otimes \Gamma^{-\frac 2N v_{N-1}}  \\ 
		\notag & + \cN \otimes \cN\sum_{i=1}^{N-2} c_{i} \Gamma^{\frac 2N v_i} \otimes \Gamma^{-\frac 2N v_i}\Bigg) \left( \cA \otimes \cA \right) d\lambda.
	\end{align}
	Under the condition that $2\pi\i \epsilon^{-1} (q^N_0-\bar{q}^N_0- \epsilon \frac2N ) \in 2\pi \I \mathbb{Z}$  this expression is equal to~\eqref{A-HQE}, which proves the invariance under the action of $\gamma_N$.
\end{proof}

Lemma~\ref{lem:SingleValued} implies that expression~\eqref{A-HQE}, restricted to  $q^N_0-\bar{q}^N_0 \in \epsilon (\mathbb{Z}+\frac 2N)$, is a formal power series in the variables $q^\alpha_\ell+\delta^{a}_{N-1}\delta^1_\ell$ and $\bar q^\alpha_\ell+\delta^{a}_{N-1}\delta^1_\ell$ for $\alpha=1,\dots,N,~\ell \geq0$ and a Laurent series in $\epsilon$ whose coefficients   are rational functions of $\lambda$ with possible poles at the points $\lambda=u^1,\dots,u^N,\infty$. 

\begin{definition} 
We say that the ancestor potential $\cA$ satisfies the \emph{ancestor Hirota quadratic equations} for the set $O$ if the aforementioned dependence on $\lambda$ is polynomial, that is, if there are no poles at $\lambda=u^1,\dots,u^N$.
\end{definition}

\subsection{Proof of the ancestor Hirota equations} 

\begin{theorem}
\label{thm:HirotaForAncestors} 
	The  ancestor potential $\cA$ satisfies the ancestor Hirota quadratic equations.
\end{theorem}

\begin{proof} 
Let us prove that~\eqref{A-HQE} is regular at $\lambda\to u^1$ (the cases of $\lambda\to u^2,\dots,u^{N-1}$ are completely analogous, and $\lambda\to u^{N}$ is just a bit more special and is discussed below). The period vector $I^{(\ell)}_{\frac 2N v_i}$ is holomorphic at $\lambda\to u^1$ for $i\not=0,1$, $\ell\in \Z$. 
Therefore,~\eqref{A-HQE} is regular at $\lambda\to u^1$ if and only if the following expression
\begin{align}
(\cN \otimes \cN)
\left(c_{0} \Gamma^{\frac 2N v_0} \otimes \Gamma^{-\frac 2N v_0} +c_{1} \Gamma^{\frac 2N v_1} \otimes \Gamma^{-\frac 2N v_1}\right) \left( \cA \otimes \cA \right) d\lambda,
\end{align}
which is single valued for $\lambda$ near  $u^1$, is also regular at $\lambda\to u^1$. Since the factor $\cN \otimes \cN$ doesn't affect the regularity at $\lambda\to u^1$, it is sufficient to prove the regularity of 
\begin{align}
\label{expr:regular-at-u1-1}
& \left(  c_{0} \Gamma^{\frac 2N v_0} \otimes \Gamma^{-\frac 2N v_0} + c_{1} \Gamma^{\frac 2N v_1} \otimes \Gamma^{-\frac 2N v_1}\right) \left( \cA \otimes \cA \right) d\lambda
\end{align}
in a neighborhood of $u^1$ where it is still single valued.

Since $\frac 2N v_0=-e_1+v$, $\frac 2N v_1=e_1+v$, where 
\begin{align}
	v=\frac{2-N}{N}e_1 - \frac 2N\sum_{j=2}^{N-1} (j-N)e_j.
\end{align}
Notice that $<e_1,v>=0$ implies that $I_v^{(\ell)}$ is holomorphic near $u^1$. By a computation similar to \cite{Giv-An}, we have:
\begin{align}
	\Gamma^{v\pm e_1} & = e^{\pm\int_{u^1}^\lambda \cW_{v,e_1}d\rho} \Gamma^v\Gamma^{\pm e_1}, & 	\Gamma^{-v\pm e_1} = e^{\mp\int_{u^1}^\lambda \cW_{v,e_1}d\rho} \Gamma^{-v}\Gamma^{\pm e_1}.
\end{align}
This allows to rewrite~\eqref{expr:regular-at-u1-1} as
\begin{align}\label{expr:regular-at-u1-2}
	 & (\Gamma^{v} \otimes \Gamma^{-v}) \Big( c_0e^{-2\int_{u^1}^\lambda \cW_{v,e_1}d\rho} \Gamma^{-e_1}\otimes\Gamma^{e_1}
	 +c_1 e^{ 2\int_{u^1}^\lambda \cW_{v,e_1}d\rho} \Gamma^{e_1} \otimes \Gamma^{-e_1}\Big) \left( \cA \otimes \cA \right) d\rho.
\end{align}
Notice that $\Gamma^{v} \otimes \Gamma^{-v}$ doesn't affect the regularity at $\lambda\to u^1$, so it is sufficient to prove the regularity of
\begin{align}\label{expr:regular-at-u1-bis2}
	&  \Big( c_0e^{-2\int_{u^1}^\lambda \cW_{v,e_1}d\rho} \Gamma^{-e_1}\otimes\Gamma^{e_1}
+c_1 e^{ 2\int_{u^1}^\lambda \cW_{v,e_1}d\rho} \Gamma^{e_1} \otimes \Gamma^{-e_1}\Big) \left( \cA \otimes \cA \right) d\rho.
\end{align}
Recall that $\cA=\hat\Psi \hat R \prod_{i=1}^N\tau_{\rm KdV_i}$, and the $\hat{R}$-conjugation of the vertex operators is given by (c.f.\ \cite[Prop. 3]{Giv-An})
\begin{align}
	\Gamma^{\pm e_1} \hat\Psi \hat R = \hat\Psi \hat R e^{\frac 12 \int_{u^1}^\lambda (\cW_{e_1,e_1} - \frac 12 \frac 1{\rho-u^1})d\rho } e^{\pm\widehat{\cf_{\rm KdV}(\lambda-u^1,z)}e_1}.
\end{align}
Notice that 
\begin{align}
	\log c_{0} -2\int_{u^1}^\lambda \cW_{v,e_1}d\rho + \int_{u^1}^\lambda (\cW_{e_1,e_1} - \frac 12 \frac 1{\rho-u^1})d\rho & = C - \frac 12 \log(\lambda-u^1),\\
	 \log c_{1} +2\int_{u^1}^\lambda \cW_{v,e_1}d\rho + \int_{u^1}^\lambda (\cW_{e_1,e_1} - \frac 12 \frac 1{\rho-u^1})d\rho & = C - \frac 12 \log(\lambda-u^1)-\pi\I,
\end{align}
where 
\begin{align}
C=-\int_{\lambda_0}^\lambda \cW_{v,v}d\rho -2\int_{u^1}^{\lambda_0} \cW_{v,e_1}d\rho +\int_{u^1}^{\lambda_0} (\cW_{e_1,e_1}-\frac 12 \frac1{\rho-u^1})d\rho +\frac 12 \log(\lambda_0-u^1).	
\end{align}
Thus we can rewrite~\eqref{expr:regular-at-u1-bis2} as
\begin{align}\label{expr:regular-at-u1-3}
 & (\hat\Psi \hat R \otimes \hat\Psi \hat R ) e^C
 \Big( (\lambda-u^1)^{-\frac 12}e^{-\widehat{\cf_{\rm KdV}(\lambda-u^1,z)}e_1}\otimes e^{\widehat{\cf_{\rm KdV}(\lambda-u^1,z)}e_1}
  \\ \notag
 &
  -(\lambda-u^1)^{-\frac 12} e^{\widehat{\cf_{\rm KdV}(\lambda-u^1,z)}e_1}\otimes e^{-\widehat{\cf_{\rm KdV}(\lambda-u^1,z)}e_1}\Big) 
 \Big(\prod_{i=1}^N\tau_{\rm KdV_i} \otimes \prod_{i=1}^N\tau_{\rm KdV_i} \Big) d\lambda. 
\end{align}
In this expression the regularity at $\lambda\to u^1$ of 
\begin{align}
	& \Big( (\lambda-u^1)^{-\frac 12}e^{-\widehat{\cf_{\rm KdV}(\lambda-u^1,z)}e_1}\otimes e^{\widehat{\cf_{\rm KdV}(\lambda-u^1,z)}e_1}
	  \\ \notag
	&
	-(\lambda-u^1)^{-\frac 12} e^{\widehat{\cf_{\rm KdV}(\lambda-u^1,z)}e_1}\otimes e^{-\widehat{\cf_{\rm KdV}(\lambda-u^1,z)}e_1}\Big) 
	\Big(\tau_{\rm KdV_1} \otimes \tau_{\rm KdV_1} \Big) d\lambda
\end{align}
is equivalent to the equations of the Korteweg--de Vries hierarchy for $\tau_{\rm KdV_1}$, and all other terms in~\eqref{expr:regular-at-u1-3} are regular at $\lambda\to u^1$.

Exactly the same argument proves regularity of~\eqref{A-HQE} at $\lambda \to u^2,\dots,u^{N-1}$. However, for $\lambda\to u^N$ we need some additional argument, where we have to use the operator $\cN$ and the assumption  $ q^N_0-\bar{q}^N_0 \in \epsilon (\mathbb{Z}+\frac 2N)$. 

Note that~\eqref{A-HQE} is regular at $\lambda\to u^N$ if the following expression
\begin{align} \label{eq:uNreg-1}
	& (\cN \otimes \cN)
	\Big(c_{0} \Gamma^{\frac 2N v_0} \otimes \Gamma^{-\frac 2N v_0} \\ \notag & \qquad +c_{N-1} \Gamma^{\frac 2N v_{N-1}} \otimes \Gamma^{-\frac 2N v_{N-1}}\Big) \left( \cA \otimes \cA \right) d\lambda \Big\vert_{q^N_0-\bar{q}^N_0 \in \epsilon (\mathbb{Z}+\frac 2N)},
\end{align}
in a neighborhood of $u^N$ where it is still single valued.
 Rewrite $\frac 2N v_0, \frac 2N v_{N-1}$ as $\frac 2N v_0=e_N+v-w$, $\frac 2N v_{N-1}=-e_N +v +w$, where 
\begin{align}
	v=\frac 1N\sum_{i=1}^{N-1} (2i-N)e_i.
\end{align}
Notice that $<e_N,v>=0$ implies that $I^{(\ell)}_v$ is holomorphic at $u^N$. We use the argument with the Baker-Campbell-Hausdorff formula as in Equation~\eqref{eq:BCHfw} and the following formulas
\begin{align}
	\Gamma^{v\pm e_1} & = e^{\pm\int_{u^1}^\lambda \cW_{v,e_1}d\lambda} \Gamma^v\Gamma^{\pm e_1}, & 	\Gamma^{-v\pm e_1} = e^{\mp\int_{u^1}^\lambda \cW_{v,e_1}d\lambda} \Gamma^{-v}\Gamma^{\pm e_1},
\end{align}
in order to rewrite~\eqref{eq:uNreg-1} as
\begin{align}\label{eq:uNreg-2}
	& (\cN \otimes \cN)
	\Big((e^{-\hat \cf_w}\otimes e^{\hat \cf_w}) c_{0}e^{ 2\int_{u^N}^\lambda \cW_{v,e_N}d\lambda} (\Gamma^{v} \otimes \Gamma^{-v}) (\Gamma^{e_N} \otimes \Gamma^{-e_N}) \\ \notag & +(e^{\hat \cf_w}\otimes e^{-\hat \cf_w})c_{N-1} e^{ -2\int_{u^N}^\lambda \cW_{v,e_N}d\lambda} (\Gamma^{v} \otimes \Gamma^{-v}) (\Gamma^{-e_N} \otimes \Gamma^{e_N}) \Big)
	\\ \notag &
	 \left( \cA \otimes \cA \right) d\lambda \Big\vert_{q^N_0-\bar{q}^N_0 \in \epsilon (\mathbb{Z}+\frac 2N)}.
\end{align}
We commute  $(\cN \otimes \cN)$ with $(e^{-\hat \cf_w}\otimes e^{\hat \cf_w})$ in the first summand and with $(e^{\hat \cf_w}\otimes e^{-\hat \cf_w})$ in the second and rewrite~\eqref{eq:uNreg-2} as
\begin{align}\label{eq:uNreg-3}
	& 
	\Big(e^{-\frac{\pi\i}\epsilon(q^n_0-\bar q^n_0)} (\cN \otimes \cN) c_{0}e^{ 2\int_{u^N}^\lambda \cW_{v,e_N}d\lambda} (\Gamma^{v} \otimes \Gamma^{-v}) (\Gamma^{e_N} \otimes \Gamma^{-e_N}) \\ \notag & +e^{\frac{\pi\i}\epsilon(q^n_0-\bar q^n_0)} (\cN \otimes \cN) c_{N-1} e^{ -2\int_{u^N}^\lambda \cW_{v,e_N}d\lambda} (\Gamma^{v} \otimes \Gamma^{-v}) (\Gamma^{-e_N} \otimes \Gamma^{e_N}) \Big)
	\\ \notag &
	 \left( \cA \otimes \cA \right) d\lambda \Big\vert_{q^N_0-\bar{q}^N_0 \in \epsilon (\mathbb{Z}+\frac 2N)}.
\end{align}
Substituting $q^N_0-\bar{q}^N_0 \in \epsilon (\mathbb{Z}+\frac 2N)$ in the first factors in both summands, we obtain 
\begin{align}\label{eq:uNreg-4}
	& 
	\Big(e^{-\frac{2\pi\i}N} (\cN \otimes \cN) c_{0}e^{ 2\int_{u^N}^\lambda \cW_{v,e_N}d\lambda} (\Gamma^{v} \otimes \Gamma^{-v}) (\Gamma^{e_N} \otimes \Gamma^{-e_N}) \\ \notag & +e^{\frac{2\pi\i}N} (\cN \otimes \cN) c_{N-1} e^{ -2\int_{u^N}^\lambda \cW_{v,e_N}d\lambda} (\Gamma^{v} \otimes \Gamma^{-v}) (\Gamma^{-e_N} \otimes \Gamma^{e_N}) \Big)
	\left( \cA \otimes \cA \right) d\lambda \Big\vert_{q^N_0-\bar{q}^N_0 \in \epsilon (\mathbb{Z}+\frac 2N)}
	\\ \notag &
	= e^{\frac {2\pi\i}N}(\cN \otimes \cN)(\Gamma^{v} \otimes \Gamma^{-v})
	\Big(
	c_{N-1} e^{ -2\int_{u^N}^\lambda \cW_{v,e_N}d\lambda} (\Gamma^{-e_N} \otimes \Gamma^{e_N}) 
	\\ \notag & \quad 
	+ c_{0,1}e^{ 2\int_{u^N}^\lambda \cW_{v,e_N}d\lambda}  (\Gamma^{e_N} \otimes \Gamma^{-e_N})
	\Big)
	\left( \cA \otimes \cA \right) d\lambda \Big\vert_{q^N_0-\bar{q}^N_0 \in \epsilon (\mathbb{Z}+\frac 2N)}. 
\end{align}
In order to prove the regularity of the latter expression at $\lambda\to u^N$, it is sufficient to prove the regularity of
\begin{align}
		\Big(
	c_{N-1} e^{ -2\int_{u^N}^\lambda \cW_{v,e_N}d\lambda} (\Gamma^{-e_N} \otimes \Gamma^{e_N}) 
	+ c_{0,1}e^{ 2\int_{u^N}^\lambda \cW_{v,e_N}d\lambda}  (\Gamma^{e_N} \otimes \Gamma^{-e_N})
	\Big)
	\left( \cA \otimes \cA \right) d\lambda,
\end{align}
which follows from exactly the same argument as the regularity of~\eqref{expr:regular-at-u1-bis2}.
\end{proof}

\subsection{Definition of descendant Hirota equations} 
Let us define
\begin{equation} \label{eq:def-N-infinity}
	\cN_\infty \coloneqq \exp\left( - \sum_{(j,\ell)\not= (N,0)} \frac{(I^{(-\ell-1)}_{w,\infty})_j q^j_\ell }{(I^{(-1)}_{w,\infty})_N}  \frac{\partial }{\partial q^N_0}  \right) = \exp\left( -
	\sum_{\ell\geq 1} \frac{\lambda^{\ell}}{\ell!}q^N_\ell \frac{\partial }{\partial q_0^N} \right)\, 
\end{equation}
so that that the following identity holds
\begin{equation*}
\cN_\infty \Gamma_\infty^{\pm 2 w} = e^{\mp 2\pi \I \frac{q_0^N}{\epsilon} } \cN_\infty .
\end{equation*}
We also define
\begin{align} \label{eq:CoefDefInfinity}
	c_{0}^{\infty} & \coloneqq -\lambda^{-2+\frac 2N},
	\\ \notag
	c_{i}^{\infty}  &\coloneqq \frac{\lambda^{-1-\frac 1{N-1}+\frac 2N}}{N-1}e^{2\pi\i (i-1)(-1-\frac 1{N-1}+\frac 2N)}, & i=1,\dots,N-1,
\end{align}
where we use the principal branch of the logarithm $\Log\lambda$ for the definition of the fractional degrees of $\lambda$.
%
%
Recall $\Gamma^{\frac 2N v_i}_\infty = e^{\widehat{\cf_{\frac 2N v_i,\infty}}} = e^{\widehat{(\cf_{ \frac 2N v_i,\infty})_-}} e^{\widehat{(\cf_{ \frac 2N v_i,\infty})_+}} $. 
Consider the Hirota one-form 
\begin{align} \label{eq:ExprHirota-1}
\omega_\infty =	\cN_\infty \otimes \cN_\infty \Bigg(
	\sum_{i=0}^{N-1} 	c_{i}^{\infty}\Gamma^{\frac 2N v_i}_\infty\otimes \Gamma^{-\frac 2N v_i}_\infty
	\Bigg) \cD\otimes \cD \ d\lambda,
\end{align}
where the two copies of $\cD$ depend on two different sets of variables, $\{q^\alpha_\ell\}$ and $\{\bar q^\alpha_\ell\}$ respectively. 

The vertex operators $\Gamma_\infty^{\pm \frac2N v_i}$ are formal asymptotic series for $\lambda \sim \infty$, namely combinations of logarithm and fractional powers of $\lambda$. However, because the triviality of its monodromy at infinity, the one-form $\omega_\infty$ turns out to be single-valued as proved in the following lemma.
\begin{lemma}\label{lem:descsinglevalues} For $q^N_0-\bar q^N_0\in \epsilon (\Z+\frac 2N)$ the Hirota one-form $\omega_\infty$ is single valued in $\lambda$.
\end{lemma}

\begin{proof}
Recall explicit formulas in Section~\ref{sec:asymptotics-and-monodromy-infty}. Note also that $\gamma_\infty c_{0}^\infty = c_{0}^\infty e^{\frac 2N 2\pi\i}$, $\gamma_\infty c_{i}^\infty = c_{i+1}^\infty$ for $i=1,\dots,N-2$, and $\gamma_\infty c^\infty_{N-1}=c_1^\infty e^{-\frac 2N 2\pi\i}$. Moreover, 
\begin{align}
\gamma_\infty (\Gamma_\infty^{\frac 2N v_0} \otimes \Gamma_\infty^{-\frac 2N v_0}) & = (\Gamma_\infty^{-2w} \otimes \Gamma_\infty^{2w})(\Gamma_\infty^{\frac 2N v_0} \otimes \Gamma_\infty^{-\frac 2N v_0}), 
\\ \notag 
\gamma_\infty (\Gamma_\infty^{\frac 2N v_i} \otimes \Gamma_\infty^{-\frac 2N v_i} )& = \Gamma_\infty^{\frac 2N v_{i+1}} \otimes \Gamma_\infty^{-\frac 2N v_{i+1}}, & i=1,\dots,N-2,
\\ \notag 
\gamma_\infty (\Gamma_\infty^{\frac 2N v_{N-1}} \otimes \Gamma_\infty^{-\frac 2N v_{N-1}}) & = (\Gamma_\infty^{2w} \otimes \Gamma_\infty^{-2w})(\Gamma_\infty^{\frac 2N v_1} \otimes \Gamma_\infty^{-\frac 2N v_1}), 
\end{align} 
Thus we have:
	\begin{align} \label{eq:ExprHirota-gammainfty}
		& \gamma_\infty\left[ \cN_\infty \otimes \cN_\infty \Bigg(
		\sum_{i=0}^{N-1} 	c_{i}^{\infty}\Gamma^{\frac 2N v_i}_\infty\otimes \Gamma^{-\frac 2N v_i}_\infty
		\Bigg) \cD\otimes \cD d\lambda \Big|_{q^N_0-\bar q^N_0\in \epsilon (\Z+\frac 2N)}\right]
		\\ \notag
		& = \cN_\infty \otimes \cN_\infty \Bigg(
		\sum_{i=2}^{N-1} 	c_{i}^{\infty}\Gamma^{\frac 2N v_i}_\infty\otimes \Gamma^{-\frac 2N v_i}_\infty
		\Bigg) \cD\otimes \cD d\lambda \Big|_{q^N_0-\bar q^N_0\in \epsilon (\Z+\frac 2N)}
		\\ \notag & \quad
		+e^{\frac 2N 2\pi\i}( \cN_\infty \otimes \cN_\infty)(\Gamma_\infty^{-2w} \otimes \Gamma_\infty^{2w})c_{0}^\infty (\Gamma_\infty^{\frac 2N v_0} \otimes \Gamma_\infty^{-\frac 2N v_0}) \cD\otimes \cD d\lambda \Big|_{q^N_0-\bar q^N_0\in \epsilon (\Z+\frac 2N)}
		\\ \notag & \quad
		+ e^{-\frac 2N 2\pi\i}( \cN_\infty \otimes \cN_\infty)(\Gamma_\infty^{2w} \otimes \Gamma_\infty^{-2w})c_{1}^\infty (\Gamma_\infty^{\frac 2N v_1} \otimes \Gamma_\infty^{-\frac 2N v_1}) \cD\otimes \cD d\lambda \Big|_{q^N_0-\bar q^N_0\in \epsilon (\Z+\frac 2N)}.
	\end{align}
Note that 
\begin{align}
	e^{\pm \frac 2N 2\pi\i}( \cN_\infty \otimes \cN_\infty)(\Gamma_\infty^{\mp 2w} \otimes \Gamma_\infty^{\pm 2w}) = e^{\pm \frac 2N 2\pi\i \mp \frac{2\pi\i}{\epsilon} (q^N_0-\bar q^N_0)}( \cN_\infty \otimes \cN_\infty),
\end{align}
hence, taking into account the restriction to $q^N_0-\bar q^N_0\in \epsilon (\Z+\frac 2N)$, we see that the right hand side of Equation~\eqref{eq:ExprHirota-gammainfty} is indeed equal to
\begin{align}\label{eq:DescendantHirota}
	\cN_\infty \otimes \cN_\infty \Bigg(
	\sum_{i=0}^{N-1} 	c_{i}^{\infty}\Gamma^{\frac 2N v_i}_\infty\otimes \Gamma^{-\frac 2N v_i}_\infty
	\Bigg) \cD\otimes \cD d\lambda \Big|_{q^N_0-\bar q^N_0\in \epsilon (\Z+\frac 2N)}.
\end{align}
\end{proof}

\begin{definition}
\label{def:DescHirota}
We say that the descendant potential $\cD$ satisfies the \emph{descendant Hirota quadratic equations} if the Hirota one-form $\omega_\infty$ is regular at $\lambda \sim \infty$.
\end{definition}

Regularity of $\omega_\infty$ means that the coefficients of the negative powers of $\lambda$ and of the monomials in $y^\alpha_d$ vanish, after the change of variables $x^\alpha_d = \frac12 ( q^\alpha_d + \bar q^\alpha_d ) + \delta^\alpha_{N-1} \delta_d^1$ and $y^\alpha_d = \frac12 ( q^\alpha_d - \bar q^\alpha_d )$, with $q^N_0 - \bar q^N_0 \in \epsilon( \Z + \frac2N )$.

\subsection{Proof of the descendant Hirota equations} 

The goal of this section is to prove 

\begin{theorem} \label{thm:DescHirota}
	The descendant potential $\cD$ satisfies the descendant Hirota quadratic equations.
\end{theorem}

\begin{proof}
The descendant potential is given by $\cD = C\hat{S}^{-1}\cA$, where the latter three factors are computed at the same point of $M$. Though $\cD$ does not depend on the choice of this point, it is convenient to use the special one, $t_{sp}$, then in particular $C=1$, and we can use the results on the ancestor Hirota quadratic equations for $\cA$ proved in Section~\ref{sec:HirotaAncestor}. So, all computations below are made specifically at the special point. 

\subsubsection*{Coefficients and conjugation of $\hat S$ and $\Gamma_\infty$}

Using Proposition~\ref{prop:conjugationS}, we have:
\begin{align}
	(\hat{S}\otimes \hat{S} )(\Gamma^{\frac 2N v_i}_\infty\otimes \Gamma^{-\frac 2N v_i}_\infty) (\hat{S}^{-1}\otimes \hat{S}^{-1}) = e^{\int_\lambda^\infty (\cW_{\frac 2N v_i,\frac 2N v_i} - \cW^\infty_{\frac 2N v_i,\frac 2N v_i})d\rho} \Gamma^{\frac 2N v_i}\otimes \Gamma^{-\frac 2N v_i}.
\end{align}
Let us prove that for $i=0,\dots,N-1$ 
\begin{align} \label{eq:Coef1}
	c_i^\infty e^{\int_\lambda^\infty (\cW_{\frac 2N v_i,\frac 2N v_i} - \cW^\infty_{\frac 2N v_i,\frac 2N v_i})d\rho} = c_i\cdot F,  
\end{align}
for some common factor $F$ defined below. 
Recall an explicit computation in Lemma~\ref{lem:coefficients}. We have:
\begin{align}\label{eq:Coef2}
	\log c_i & = - \int_{p_0(\lambda_0)}^{p_i(\lambda)} \left(d\log \left(\frac{d f(p)}{d p}\right) + \frac 2N \frac{dp}{p} \right)
	\\ \notag & = - \log \left(\frac{d f(p)}{d p}\right) \big|_{p=p_i(\lambda)} - \frac 2N \log p  \big|_{p=p_i(\lambda)}  \\ \notag & \quad +  \log \left(\frac{d f(p)}{d p}\right) \big|_{p=p_0(\lambda_0)} + \frac 2N \log p  \big|_{p=p_0(\lambda_0)}.
\end{align}
Define $-\log F$ as the sum of the last two summands. For the first two summands, using that $f(p)=p^{N-1}+p^{-1}$, we observe that for $\lambda\to \infty$ we have
\begin{align}
	 & \log \left(\frac{d f(p)}{d p}\right) \big|_{p=p_i(\lambda)} + \frac 2N \log p  \big|_{p=p_i(\lambda)} 
	 \\ \notag &
	 \sim 
	 \begin{cases}
	 \Log(N-1) +\left(1+\frac{1}{N-1}-\frac 2N \right)(\Log\lambda + 2\pi\i (i-1)), & i=1,\dots,N-1; \\
	 \Log(-1) + \left(2-\frac 2N\right) \Log\lambda, & i=0.
	 \end{cases}
\end{align}
Therefore, 
\begin{align} \label{eq:Coef3}
	& \int^\infty_\lambda (\cW_{\frac 2N v_i,\frac 2N v_i} - \cW^\infty_{\frac 2N v_i,\frac 2N v_i})d\lambda = 
	\\ \notag & 
	\begin{cases}
		- \log \left(\frac{d f(p)}{d p}\right) \big|_{p=p_i(\lambda)} - \frac 2N \log p  \big|_{p=p_i(\lambda)} & \\ \qquad +  \Log(N-1) +\left(1+\frac{1}{N-1}-\frac 2N \right)(\Log\lambda + 2\pi\i (i-1)), & i=1,\dots,N-1; \\
			- \log \left(\frac{d f(p)}{d p}\right) \big|_{p=p_0(\lambda)} - \frac 2N \log p  \big|_{p=p_0(\lambda)} & \\ \qquad + \Log(-1) + \left(2-\frac 2N\right) \Log\lambda, & i=0.
	\end{cases}
\end{align}
Now Equation~\eqref{eq:Coef1} follows directly by substitution of Equations~\eqref{eq:CoefDefInfinity} and~\eqref{eq:Coef2} on the left hand side and  and~\eqref{eq:Coef3} on the right hand side. Thus we have the following equality of asymptotic series:
\begin{align}\label{eq:DescHirProof1}
	(\hat{S}\otimes \hat{S} )(c_i^\infty \Gamma^{\frac 2N v_i}_\infty\otimes \Gamma^{-\frac 2N v_i}_\infty) (\hat{S}^{-1}\otimes \hat{S}^{-1}) =F \cdot c_i \Gamma^{\frac 2N v_i}\otimes \Gamma^{-\frac 2N v_i},
\end{align}
$i=0,\dots,N-1$.

\subsubsection*{Conjugation of $\hat S$ and $\cN_\infty$} 

We recall~\cite[Lemma 39]{carletHigherGeneraCatalan2021}. It is a universal statement whose proof doesn't use any specifics of the underlying Frobenius manifold (except for the properties  $I_{w,\infty}^\ell=0$ and $I_{w}^\ell=0$ for $\ell\geq 0$ that we do have in our case). In our case, exactly the same argument implies that we have the following equality of asymptotic series:
\begin{align}\label{eq:DescHirProof2}
	\cN_\infty \hat S^{-1} = Q T \cN,
\end{align}
where $Q$ is an exponential of a linear combination of terms $\epsilon^{-2} q^i_\ell q^j_m$ with the coefficients polynomial in $\lambda$, and $T$ is an exponential of a linear vector field in $q^i_\ell$ with constant coefficients that does not contain differentiation $\partial /\partial q^N_0$.

\subsubsection*{Final steps of the proof}

Equations~\eqref{eq:DescHirProof1} and~\eqref{eq:DescHirProof2} imply that we can rewrite~\eqref{eq:DescendantHirota} as the asymptotic series expansion at $\lambda \to \infty$ of the following expression:
\begin{align}
	F(Q\otimes Q)(T\otimes T) (\cN\otimes \cN) \left(\sum_{i=0}^{N-1} c_{i}\Gamma^{\frac 2N v_i} \otimes \Gamma^{-\frac 2Nv_i}\right) \left( \cA \otimes \cA \right) d\lambda \Big|_{q^N_0-\bar q^N_0\in \epsilon (\Z+\frac 2N)}.
\end{align}
Since the operator $F(Q\otimes Q)(T\otimes T)$ does not contain derivatives with respect to $q^N_0$ and $\bar q^N_0$, we can rewrite the above expression as 
\begin{align} \label{eq:DescHirProof3}
		 \Big(F(Q\otimes Q)&(T\otimes T) \Big)\Big|_{q^N_0-\bar q^N_0\in \epsilon (\Z+\frac 2N)} \times\\ \notag &\times\left((\cN\otimes \cN) \left(\sum_{i=0}^{N-1} c_{i}\Gamma^{\frac 2N v_i} \otimes \Gamma^{-\frac 2Nv_i}\right) \left( \cA \otimes \cA \right) d\lambda\right) \Big|_{q^N_0-\bar q^N_0\in \epsilon (\Z+\frac 2N)}.
\end{align}
Remark that the operator $F(Q\otimes Q)(T\otimes T)$ restricted to $q^N_0-\bar q^N_0\in \epsilon (\Z+\frac 2N)$ is an invertible operator that preserves the polynomiality in $\lambda$, and that by Theorem~\ref{thm:HirotaForAncestors}  the second line of~\eqref{eq:DescHirProof3} is polynomial in $\lambda$ in the sense of Definition~\ref{def:DescHirota}. This implies that~\eqref{eq:DescendantHirota} is polynomial in $\lambda$ in the sense of Definition~\ref{def:DescHirota}: this finishes the proof that the descendant potential $\cD$ satisfies the descendant Hirota quadratic equations. 
\end{proof}

\subsection{Explicit form of the descendant Hirota equations} 

In this section we use the formulas for $\cf_{ \frac 2N v_i,\infty}$, $i=0,\dots,N-1$, computed above, as well as the quantization rules, in order to rewrite the descendant Hirota quadratic equations more explicitly. 

\begin{corollary}\label{cor:explicitform} For any $k\in\Z$ and for any $n\geq 0$ the descendant potential $\cD$ satisfies the following equations:
\begin{align}
\label{Hirota}
	0 & = \res_{\lambda = \infty} \lambda^{n-1} d\lambda \Bigg[ -\lambda^{-k} \times
	\\ \notag & \qquad 
	\exp\left(-\frac 1\epsilon \sum_{\ell=0}^\infty \frac 12 \frac{\lambda^{\ell+1}}{(\ell+1)!} (q^1_\ell -\bar q^1_\ell)
	+\frac 1\epsilon \sum_{\ell=1}^\infty \frac{\lambda^{\ell}}{\ell!} \ch(\ell) (q^N_\ell-\bar q^N_\ell) \right)\times
	\\ \notag & \qquad 
	\cD\left(\left\{q^1_\ell + \epsilon \frac{\ell!}{\lambda^{\ell+1}} \right\}_{\ell\geq 0},\  \{q^\alpha_\ell\}_{\substack{\ell\geq 0 \\ \alpha=2,\dots,N-1 }}
	,\ q^N_0+\frac{\epsilon}{2}-
	\sum_{\ell\geq 1} \frac{\lambda^{\ell}}{\ell!}q^N_\ell,\ \{q^N_\ell\}_{\ell\geq 1}\right) \times 
	\\ \notag &\qquad 
	\cD\left(\left\{\bar q^1_\ell - \epsilon \frac{\ell!}{\lambda^{\ell+1}} \right\}_{\ell\geq 0},\  \{\bar q^\alpha_\ell\}_{\substack{\ell\geq 0 \\ \alpha=2,\dots,N-1 }}
	,\ \bar q^N_0-\frac{\epsilon}{2} -
	\sum_{\ell\geq 1} \frac{\lambda^{\ell}}{\ell!}\bar q^N_\ell,\ \{\bar q^N_\ell\}_{\ell\geq 1}\right) +
	\\ \notag &
	+\sum_{i=1}^{N-1} 	\frac{1}{N-1} \lambda^{\frac{k}{N-1}} \exp\left(2\pi\i(i-1)\frac k{N-1}\right)\times 
	\\ \notag 
	& \qquad 
	\exp\Bigg( \frac 1\epsilon \sum_{\ell=0}^\infty \frac{1}{N} \frac{\lambda^{\ell+1}}{(\ell+1)!} (q^1_\ell - \bar q^1_\ell)
	\\ \notag &
	\qquad \qquad +
	\frac 1\epsilon \sum_{\alpha=2}^{N-1} \sum_{\ell=0}^\infty  \frac{\lambda^{\frac{N-\alpha}{N-1}+\ell}}{(N-1)\left(\frac{N-\alpha}{N-1}\right)_{(\ell+1)}} \exp\Big({2\pi\i(i-1)\frac{1-\alpha}{N-1}}\Big)
	(q^\alpha_\ell - \bar q^\alpha_\ell )
	\\ \notag &
	\qquad \qquad -
	\frac 1\epsilon \sum_{\ell=1}^\infty \frac{ \ch(\ell)}{N-1} \frac{\lambda^{\ell}}{\ell!}  (q^N_\ell-\bar q^N_\ell) \Bigg)\times
	\\ \notag & \qquad 
	\cD\Bigg(\left\{q^1_\ell - \epsilon \frac{1}{N-1}\frac{\ell!}{\lambda^{\ell+1}} \right\}_{\ell\geq 0},
	\\ \notag & \qquad \qquad 
	\left\{q^\alpha_\ell - \epsilon  \left(\frac{N-\alpha}{N-1}\right)_\ell \lambda^{\frac{\alpha-N}{N-1}-\ell} \exp\Big({2\pi\i(i-1)\frac{\alpha-N}{N-1}}\Big) \right\}_{\substack{\ell\geq 0 \\ \alpha=2,\dots,N-1 }}
	,
	\\ \notag &  \qquad \qquad 
	q^N_0-\epsilon\frac{1}{N}-
	\sum_{\ell\geq 1} \frac{\lambda^{\ell}}{\ell!}q^N_\ell,\ \{q^N_\ell\}_{\ell\geq 1}\Bigg) \times 
	\\ \notag & \qquad 
	\cD\Bigg(\left\{\bar q^1_\ell + \epsilon \frac{1}{N-1}\frac{\ell!}{\lambda^{\ell+1}} \right\}_{\ell\geq 0},
	\\ \notag & \qquad \qquad 
	\left\{\bar q^\alpha_\ell + \epsilon  \left(\frac{N-\alpha}{N-1}\right)_\ell \lambda^{\frac{\alpha-N}{N-1}-\ell} \exp\Big({2\pi\i(i-1)\frac{\alpha-N}{N-1}}\Big) \right\}_{\substack{\ell\geq 0 \\ \alpha=2,\dots,N-1 }}
	,
	\\ \notag &  \qquad \qquad 
	\bar q^N_0+\epsilon\frac{1}{N}-
	\sum_{\ell\geq 1} \frac{\lambda^{\ell}}{\ell!}\bar q^N_\ell,\ \{\bar q^N_\ell\}_{\ell\geq 1}\Bigg)
	\Bigg],
\end{align}
where $q^N_0-\bar q^N_0 = \epsilon (k-1+\frac 2N)$.
\end{corollary}

\begin{proof}
Let us compute the explicit expression of the Hirota one-form $\omega_\infty$. 
From Equations~\eqref{eq:finfty-0} and~\eqref{eq:finfty-i} and the quantization rules we see that
	\begin{align}
		\hat\cf_{ \frac 2N v_0,\infty} & = 
		-\frac 1\epsilon \sum_{\ell=0}^\infty \frac 12 \frac{\lambda^{\ell+1}}{(\ell+1)!} q^1_\ell 
		-\frac 1\epsilon \sum_{\ell=0}^\infty \frac{\lambda^{\ell}}{\ell!} (\Log \lambda - \ch(\ell)) q^N_\ell \\ \notag & \quad 
		+\epsilon \sum_{\ell=0}^\infty \frac{\ell!}{\lambda^{\ell+1}} \frac{\partial}{\partial q^1_\ell} + \frac\epsilon 2 \frac{\partial}{\partial q^N_0},
	\end{align}
and for $i=1,\dots,N-1$ we have
\begin{align}
	\hat \cf_{ \frac 2N v_i,\infty} & = 
	\frac 1\epsilon \sum_{\ell=0}^\infty \frac{1}{N} \frac{\lambda^{\ell+1}}{(\ell+1)!} q^1_\ell 
	\\ \notag &
	\quad +
	\frac 1\epsilon \sum_{\alpha=2}^{N-1} \sum_{\ell=0}^\infty  \frac{\lambda^{\frac{N-\alpha}{N-1}+\ell}}{(N-1)\left(\frac{N-\alpha}{N-1}\right)_{(\ell+1)}} \exp\Big({2\pi\i(i-1)\frac{1-\alpha}{N-1}}\Big)
	q^\alpha_\ell 
	\\ \notag &
	\quad +
	\frac 1\epsilon \sum_{\ell=0}^\infty \frac{1}{N-1} \frac{\lambda^{\ell}}{\ell!} (\Log\lambda -\ch(\ell) + 2\pi\i(i-1)) q^N_\ell
	\\ \notag & \quad
	- \epsilon \sum_{\ell=0}^\infty \frac{1}{N-1} \frac{\ell!}{\lambda^{\ell+1}} \frac{\partial}{\partial q^1_\ell}	
	\\ \notag & \quad
	- \epsilon \sum_{\alpha=2}^{N-1} \sum_{\ell=0}^\infty 
	\left(\frac{N-\alpha}{N-1}\right)_{\ell} \lambda^{\frac{\alpha-N}{N-1}-\ell} \exp\Big({2\pi\i(i-1)\frac{\alpha-N}{N-1}}\Big) 
	 \frac{\partial}{\partial q^\alpha_\ell}	
	\\ \notag & \quad
	- \epsilon \frac{1}{N} \frac{\partial}{\partial q^N_0}.
\end{align}

Notice that $\cN_\infty$ and $ \left( \hat\cf_{\frac2N v_i, \infty} \right)_+$ commute, while
\begin{equation}
\cN_\infty \left( \hat\cf_{\frac2N v_i, \infty} \right)_- = \left( \tilde\cf_{\frac2N v_i, \infty} \right)_- \cN_\infty
\end{equation}
where $ \left( \tilde\cf_{\frac2N v_i, \infty} \right)_-$ equals $\left( \hat\cf_{\frac2N v_i, \infty} \right)_-$ with $q_0^N$ shifted by $-\sum_{\ell \geq 1} \frac{\lambda^\ell}{\ell !} q_\ell^N$. This in particular kills all the terms containing $\Log \lambda$, but those containinig $\Log \lambda$ and $q_0^N$. 
Explicitly 
\begin{equation}
\left( \tilde\cf_{\frac2N v_0, \infty} \right)_- = 
- \frac1\epsilon \Log \lambda \ q_0^N
+\frac 1\epsilon \sum_{\ell=1}^\infty \frac{\lambda^{\ell}}{\ell!} \ch(\ell) q^N_\ell
-\frac 1\epsilon \sum_{\ell=0}^\infty \frac 12 \frac{\lambda^{\ell+1}}{(\ell+1)!} q^1_\ell  
\end{equation}
and for $i=1,\dots,N-1$ 
\begin{align}
\left( \tilde\cf_{\frac2N v_i, \infty} \right)_- &=
\frac1{\epsilon (N-1)} \left( \Log \lambda + 2\pi \i (i-1)  \right) q_0^N  \\ 
&\quad +\frac 1\epsilon \sum_{\alpha=2}^{N-1} \sum_{\ell=0}^\infty  \frac{\lambda^{\frac{N-\alpha}{N-1}+\ell}}{(N-1)\left(\frac{N-\alpha}{N-1}\right)_{(\ell+1)}} \exp\Big({2\pi\i(i-1)\frac{1-\alpha}{N-1}}\Big)
	q^\alpha_\ell \\
&\quad + \frac 1\epsilon \sum_{\ell=0}^\infty \frac{1}{N} \frac{\lambda^{\ell+1}}{(\ell+1)!} q^1_\ell 
- \frac 1\epsilon \sum_{\ell=1}^\infty \frac{1}{N-1} \frac{\lambda^{\ell}}{\ell!} \ch(\ell) 
q^N_\ell .
\end{align}

Therefore, setting $q^N_0-\bar q^N_0 = \epsilon (k-1+\frac 2N)$, we have for $i=0$
\begin{align} \label{eq:propHQEexplicit-2ndterm}
	& c_0^\infty (\cN_\infty \otimes \cN_\infty)(\Gamma_\infty^{\frac 2N v_0}\otimes \Gamma_\infty^{-\frac 2N v_0}) (\cD\otimes \cD) 
	= 
	\\ \notag 
	=&  -\lambda^{-k-1} \exp\left(-\frac 1\epsilon \sum_{\ell=0}^\infty \frac 12 \frac{\lambda^{\ell+1}}{(\ell+1)!} (q^1_\ell -\bar q^1_\ell)
	+\frac 1\epsilon \sum_{\ell=1}^\infty \frac{\lambda^{\ell}}{\ell!} \ch(\ell) (q^N_\ell-\bar q^N_\ell) \right)\times
	\\ \notag 
	&  
	\cD\left(\left\{q^1_\ell + \epsilon \frac{\ell!}{\lambda^{\ell+1}} \right\}_{\ell\geq 0},\  \{q^\alpha_\ell\}_{\substack{\ell\geq 0 \\ \alpha=2,\dots,N-1 }}
	,\ q^N_0+\frac{\epsilon}{2}-
	\sum_{\ell\geq 1} \frac{\lambda^{\ell}}{\ell!}q^N_\ell,\ \{q^N_\ell\}_{\ell\geq 1}\right) \times 
	\\ \notag 
	&  
	\cD\left(\left\{\bar q^1_\ell - \epsilon \frac{\ell!}{\lambda^{\ell+1}} \right\}_{\ell\geq 0},\  \{\bar q^\alpha_\ell\}_{\substack{\ell\geq 0 \\ \alpha=2,\dots,N-1 }}
	,\ \bar q^N_0-\frac{\epsilon}{2} -
	\sum_{\ell\geq 1} \frac{\lambda^{\ell}}{\ell!}\bar q^N_\ell,\ \{\bar q^N_\ell\}_{\ell\geq 1}\right) 
\end{align}
and for $i=1,\dots,N-1$
\begin{align} \label{eq:propHQEexplicit-1stterm}
	& c_i^\infty (\cN_\infty \otimes \cN_\infty)(\Gamma_\infty^{\frac 2N v_i}\otimes \Gamma_\infty^{-\frac 2N v_i}) (\cD\otimes \cD) 
	= 
		\\ \notag 
	=& 
	\frac{1}{N-1} \lambda^{-1+\frac{k}{N-1}} \exp\left(2\pi\i(i-1)\big(-1+\frac k{N-1}\big)\right)\times 
		\\ \notag 
	& 
	\exp\Bigg( \frac 1\epsilon \sum_{\ell=0}^\infty \frac{1}{N} \frac{\lambda^{\ell+1}}{(\ell+1)!} (q^1_\ell - \bar q^1_\ell)
	\\ \notag &
	\qquad +
	\frac 1\epsilon \sum_{\alpha=2}^{N-1} \sum_{\ell=0}^\infty  \frac{\lambda^{\frac{N-\alpha}{N-1}+\ell}}{(N-1)\left(\frac{N-\alpha}{N-1}\right)_{(\ell+1)}} \exp\Big({2\pi\i(i-1)\frac{1-\alpha}{N-1}}\Big)
	(q^\alpha_\ell - \bar q^\alpha_\ell )
	\\ \notag &
	\qquad -
	\frac 1\epsilon \sum_{\ell=1}^\infty \frac{ \ch(\ell)}{N-1} \frac{\lambda^{\ell}}{\ell!}  (q^N_\ell-\bar q^N_\ell) \Bigg)\times
	\\ \notag &
	\cD\Bigg(\left\{q^1_\ell - \epsilon \frac{1}{N-1}\frac{\ell!}{\lambda^{\ell+1}} \right\}_{\ell\geq 0},
	\\ \notag & \qquad
	\left\{q^\alpha_\ell - \epsilon  \left(\frac{N-\alpha}{N-1}\right)_\ell \lambda^{\frac{\alpha-N}{N-1}-\ell} \exp\Big({2\pi\i(i-1)\frac{\alpha-N}{N-1}}\Big) \right\}_{\substack{\ell\geq 0 \\ \alpha=2,\dots,N-1 }}
	,
	\\ \notag &  \qquad 
	q^N_0-\epsilon\frac{1}{N}-
	\sum_{\ell\geq 1} \frac{\lambda^{\ell}}{\ell!}q^N_\ell,\ \{q^N_\ell\}_{\ell\geq 1}\Bigg) \times 
	\\ \notag &
	\cD\Bigg(\left\{\bar q^1_\ell + \epsilon \frac{1}{N-1}\frac{\ell!}{\lambda^{\ell+1}} \right\}_{\ell\geq 0},
	\\ \notag & \qquad
	\left\{\bar q^\alpha_\ell + \epsilon  \left(\frac{N-\alpha}{N-1}\right)_\ell \lambda^{\frac{\alpha-N}{N-1}-\ell} \exp\Big({2\pi\i(i-1)\frac{\alpha-N}{N-1}}\Big) \right\}_{\substack{\ell\geq 0 \\ \alpha=2,\dots,N-1 }}
	,
	\\ \notag &  \qquad 
	\bar q^N_0+\epsilon\frac{1}{N}-
	\sum_{\ell\geq 1} \frac{\lambda^{\ell}}{\ell!}\bar q^N_\ell,\ \{\bar q^N_\ell\}_{\ell\geq 1}\Bigg).
\end{align}

The Hirota one-form $\omega_\infty$ is regular at $\lambda=\infty$ iff 
\begin{equation}
\res_{\lambda=\infty} \omega_\infty  \lambda^n d\lambda = 0 \text{ for } n \geq 0.
\end{equation} 
Substituting Equations~\eqref{eq:propHQEexplicit-1stterm} and~\eqref{eq:propHQEexplicit-2ndterm} into the previous equation implies the statement of the proposition.

Notice that we have chosen $q^N_0-\bar q^N_0 = \epsilon (k-1+\frac 2N)$ to simplify the comparison of the formulas with the ones obtained in the case $N=2$ in~\cite[Proposition~40]{carletHigherGeneraCatalan2021}). 
\end{proof}

\section{The Lax formulation}  
\label{sec:LaxFormulation}

\subsection{Lax representation with difference operators} 
\label{subsec:LaxDifference}

Since multiplication by an invertible formal power series in $\lambda$ does not affect regularity, we have that the Hirota one-form $\omega_\infty$ is regular at $\lambda=\infty$  iff the $\partial_x$-valued Hirota one-form 
\begin{equation}
\tilde{\omega}_\infty := e^{\sum_{l>0} \frac{\lambda^{l}}{l!} q^N_l \partial_x} e^{ x \frac{\partial }{\partial q_0^N}} \otimes
e^{ x \frac{\partial }{\partial q^N_0}}  e^{-\sum_{l>0} \frac{\lambda^{l}}{l!} q^N_l \partial_x} 
\ \omega_\infty
\end{equation}
is regular.
Equivalently  $\tilde\omega_\infty$ is regular iff 
\begin{equation}
\res_{\lambda=\infty} \tilde\omega_\infty  \lambda^n d\lambda = 0 \text{ for } n \geq 0.
\end{equation} 
The last formula is explicitly written as equation~\eqref{Hirota} with the insertion of the operator 
\begin{equation}
e^{\sum_{l>0} \frac{\lambda^{l}}{l!} q^N_l \partial_x} e^{ x \frac{\partial }{\partial q_0^N}} 
e^{ x \frac{\partial }{\partial \bar{q}_0^N}}  e^{-\sum_{l>0} \frac{\lambda^{l}}{l!} \bar{q}^N_l \partial_x}
\end{equation}
just after the residue. 

Notice that in the second summand in~\eqref{Hirota} we are averaging over the $N-1$-th roots of $\lambda$, therefore only integer powers of $\lambda$ are present. After substituting $\lambda$ with $\lambda^{N-1}$ we get the equivalent equation:
\begin{align}
\label{Hirota'}
	& \res_{\lambda = \infty} \lambda^{n-k-1} 
	e^{\sum_{l>0} \frac{\lambda^{l}}{l!} q^N_l \partial_x} e^{ x \frac{\partial }{\partial q_0^N}} 
e^{ x \frac{\partial }{\partial \bar{q}_0^N}}  e^{-\sum_{l>0} \frac{\lambda^{l}}{l!} \bar{q}^N_l \partial_x} \times  \\
\notag & \qquad \times  \Bigg[  
	\exp\left(-\frac 1\epsilon \sum_{\ell=0}^\infty \frac 12 \frac{\lambda^{\ell+1}}{(\ell+1)!} (q^1_\ell -\bar q^1_\ell)
	+\frac 1\epsilon \sum_{\ell=1}^\infty \frac{\lambda^{\ell}}{\ell!} \ch(\ell) (q^N_\ell-\bar q^N_\ell) \right)\times
	\\ \notag &\qquad \qquad 
	\cD\left(\left\{q^1_\ell + \epsilon \frac{\ell!}{\lambda^{\ell+1}} \right\}_{\ell\geq 0},\  \{q^\alpha_\ell\}_{\substack{\ell\geq 0 \\ \alpha=2,\dots,N-1 }}
	,\ q^N_0+\frac{\epsilon}{2}-
	\sum_{\ell\geq 1} \frac{\lambda^{\ell}}{\ell!}q^N_\ell,\ \{q^N_\ell\}_{\ell\geq 1}\right) \times 
	\\ \notag &\qquad \qquad 
	\cD\left(\left\{\bar q^1_\ell - \epsilon \frac{\ell!}{\lambda^{\ell+1}} \right\}_{\ell\geq 0},\  \{\bar q^\alpha_\ell\}_{\substack{\ell\geq 0 \\ \alpha=2,\dots,N-1 }}
	,\ \bar q^N_0-\frac{\epsilon}{2} -
	\sum_{\ell\geq 1} \frac{\lambda^{\ell}}{\ell!}\bar q^N_\ell,\ \{\bar q^N_\ell\}_{\ell\geq 1}\right) \Bigg]  d\lambda =
	\\ \notag 
&	   =
\res_{\lambda = \infty} \lambda^{(N-1)n+k-1}
e^{\sum_{l>0} \frac{\lambda^{l(N-1)}}{l!} q^N_l \partial_x} e^{ x \frac{\partial }{\partial q_0^N}} 
e^{ x \frac{\partial }{\partial \bar{q}_0^N}}  e^{-\sum_{l>0} \frac{\lambda^{l(N-1)}}{l!} \bar{q}^N_l \partial_x} \times  \\
&\qquad  \times  \Bigg[
	\exp\Bigg( \frac 1\epsilon \sum_{\ell=0}^\infty \frac{1}{N} \frac{\lambda^{(N-1)(\ell+1)}}{(\ell+1)!} (q^1_\ell - \bar q^1_\ell) \notag
	\\ \notag  &
\qquad  \qquad \qquad 	+ 
	\frac 1\epsilon \sum_{\alpha=2}^{N-1} \sum_{\ell=0}^\infty  \frac{\lambda^{{N-\alpha}+\ell (N-1)}}{(N-1)\left(\frac{N-\alpha}{N-1}\right)_{(\ell+1)}}  
	(q^\alpha_\ell - \bar q^\alpha_\ell )
	-
	\frac 1\epsilon \sum_{\ell=1}^\infty \frac{ \ch(\ell)}{N-1} \frac{\lambda^{\ell (N-1)}}{\ell!}  (q^N_\ell-\bar q^N_\ell) \Bigg)\times
	\\ \notag &
\qquad \qquad 	\cD\Bigg(\left\{q^1_\ell -  \frac{\epsilon}{N-1}\frac{\ell!}{\lambda^{(N-1)(\ell+1)}} \right\}_{\ell\geq 0},
	\left\{q^\alpha_\ell - \epsilon  \left(\frac{N-\alpha}{N-1}\right)_\ell \lambda^{\alpha-N-\ell (N-1)}   \right\}_{\substack{\ell\geq 0 \\ \alpha=2,\dots,N-1 }}
	,
	\\ \notag &  \qquad \qquad \qquad 
	q^N_0-\frac{\epsilon}{N}-
	\sum_{\ell\geq 1} \frac{\lambda^{(N-1)\ell}}{\ell!}q^N_\ell,\ \{q^N_\ell\}_{\ell\geq 1}\Bigg) \times 
	\\ \notag &
\qquad \qquad 	\cD\Bigg(\left\{\bar q^1_\ell +  \frac{\epsilon}{N-1}\frac{\ell!}{\lambda^{(\ell+1)(N-1)}} \right\}_{\ell\geq 0},
	\left\{\bar q^\alpha_\ell + \epsilon  \left(\frac{N-\alpha}{N-1}\right)_\ell \lambda^{\alpha-N-(N-1)\ell}   \right\}_{\substack{\ell\geq 0 \\ \alpha=2,\dots,N-1 }}
	,
	\\ \notag &  \qquad \qquad \qquad 
	\bar q^N_0+\frac{\epsilon}{N}-
	\sum_{\ell\geq 1} \frac{\lambda^{(N-1)\ell}}{\ell!}\bar q^N_\ell,\ \{\bar q^N_\ell\}_{\ell\geq 1}\Bigg)
	\Bigg]  d\lambda ,
\end{align}
for $k \in \Z$ and $n \geq 0$, 
where $q^N_0-\bar q^N_0 = \epsilon (k-1+\frac 2N)$.

Following a general procedure, see e.g.~\cite{carletHigherGeneraCatalan2021} and~\cite{CvdL2013}, 
we define 
\begin{equation}
\begin{aligned}
&\cW^{+1}= \cP^{+1}\cE^{+1},
\quad
&\cW^{-1}=\cE^{-1} \cP^{-1}\\
&\cW^{+0}= \cP^{+0}\cE^{+0},
\quad
&\cW^{-0}=\cE^{-0} \cP^{-0},
\end{aligned}
\end{equation}
where
\begin{equation}
\begin{aligned}
&\cE^{\pm 1}= \\
&\quad\exp\Bigg( 
	\pm \frac 1\epsilon \sum_{\alpha=1}^{N-1} \sum_{\ell=0}^\infty  
	\frac{\lambda^{N-\alpha+\ell(N-1)}}{(N-1+\delta_{\alpha 1})\left(\frac{N-\alpha}{N-1}\right)_{(\ell+1)}}  
q^\alpha_\ell  
	\pm   \sum_{\ell=1}^\infty \left( \partial_x -\frac{\ch(\ell)}{(N-1)\epsilon}\right) \frac{\lambda^{\ell(N-1)}}{\ell!}  q^N_\ell \Bigg), 
\end{aligned}
\end{equation}
\begin{equation}
\cE^{\pm 0}=  \exp\left(\mp \frac 1{2\epsilon }\sum_{\ell=0}^\infty   \frac{\lambda^{\ell+1}}{(\ell+1)!} q^1_\ell  
	\pm   \sum_{\ell=1}^\infty \frac{\lambda^{\ell}}{\ell!}
\left(\partial_x  +\frac{ \ch(\ell)}{\epsilon} \right)q^N_\ell  \right) 
\end{equation} 
and
\begin{equation}  
\label{8.9}
\cP^{\pm 1}= \frac{e^{\pm ( \hat{\cf}_{ \frac 2N v_1,\infty} )_+|_{\lambda\to \lambda^{N-1}}}\cD'}{\cD'(x\mp\frac{\epsilon}N)}, \qquad 
\cP^{ \pm 0} = \frac{e^{\pm (\hat{\cf}_{ \frac 2N v_0,\infty} )_+}\cD'}{ \cD'(x\mp  \frac{\epsilon}N)}.
\end{equation}
Here $\cD'(q,x) = \cD(q)_{| q^N_0 \to q^N_0 + x}$.

Equation~\eqref{Hirota'} is equivalent to
\begin{align}
\label{Hirota2}
  \res_{\lambda = \infty} \lambda^{(N-1)n+k}  \cW^{+ 1}(q)\cW ^{- 1}(\bar q) \frac{d \lambda}{\lambda} =
 \res_{\lambda = \infty} \lambda^{n-k}
 \cW^{+ 0}(q)\cW^{- 0}(\bar q)  \frac{d \lambda}{\lambda},
\end{align}
for $k\in\Z$ and $n\geq0$, where  $q^N_0-\bar q^N_0 = \epsilon (k-1+\frac 2N)$.

Let us convert this expression in a bilinear equation for difference operators.
Given a difference operator $A = \sum_s a_s \Lambda^s = \sum_{s}  \Lambda^s \tilde{a}_s$ the left and right symbols are respectively defined as $\sigma_l(A) = \sum_s a_s \lambda^s$ and $\sigma_r(A) = \sum_s \tilde{a}_s \lambda^s$. Recall that 
\begin{equation} \notag
\res_\lambda \sigma_l(A) \sigma_r(B) \frac{d \lambda}\lambda = \res_\Lambda AB,
\end{equation}
where $\res_\Lambda A := a_0$, for a proof see \S 3.2 in~\cite{CvdL2013}.

Let us define operators $W^{\pm 1}$ and $W^{\pm 0}$ by 
\begin{equation}  
\sigma_l(W^{+ 1})= \cW^{+ 1} , \qquad 
\sigma_r(W^{-1} ) = \cW^{-1},
\end{equation}
\begin{equation}  
\sigma_l(W^{+0})= \cW^{+0}|_{\lambda \to \lambda^{-1}  } , \qquad 
\sigma_r(W^{ -0} ) = \cW^{-0}|_{\lambda \to \lambda^{-1}  } 
\end{equation}
which implies
\begin{equation}  
\begin{aligned}
W^{+1} &= P^{+1} \times\\
&\exp\Bigg( 
	 \frac 1\epsilon \sum_{\alpha=1}^{N-1} \sum_{\ell=0}^\infty  \frac{\Lambda^{{N-\alpha}+\ell(N-1)}}
{(N-1+\delta_{\alpha 1})\left(\frac{N-\alpha}{N-1}\right)_{(\ell+1)}}  
q^\alpha_\ell  
	+   \sum_{\ell=1}^\infty \left( \partial_x -\frac{\ch(\ell)}{(N-1)\epsilon}\right) \frac{\Lambda^{\ell(N-1)}}{\ell!}  q^N_\ell \Bigg), \\
 W^{-1} &=\\  
&\exp\Bigg( 
	- \frac 1\epsilon \sum_{\alpha=1}^{N-1} \sum_{\ell=0}^\infty  \frac{\Lambda^{{N-\alpha}+\ell(N-1)}}
{(N-1+\delta_{\alpha 1})\left(\frac{N-\alpha}{N-1}\right)_{(\ell+1)}} 
- \sum_{\ell=1}^\infty \left( \partial_x -\frac{\ch(\ell)}{(N-1)\epsilon}\right) \frac{\Lambda^{\ell(N-1)}}{\ell!}  q^N_\ell \Bigg)P^{-1},\\
W^{+0} &=  P^{+0} \exp \left({ -\frac1{2\epsilon} \sum_{l \geq 0} \frac{\Lambda^{-l-1} }{(l+1)!} q^1_l  
+ \sum_{l >0} \frac{\Lambda^{-l} }{l!} (\partial_x + \frac{\ch(\ell)}{\epsilon} ) q^N_l } \right),\\
W^{-0} &=  \exp \left({ +\frac1{2\epsilon} \sum_{l \geq 0} \frac{\Lambda^{-l-1} }{(l+1)!} q^1_l  
- \sum_{l >0} \frac{\Lambda^{-l} }{l!} (\partial_x + \frac{\ch(\ell)}{\epsilon} ) q^N_l } \right) P^{-0}.
\end{aligned}
\end{equation}
Here the operators $P^{\pm 1}$ and $P^{ \pm 0}$ have been defined by 
\begin{equation}  
\sigma_l(P^{+1})= \cP^{+1} \qquad 
\sigma_r(P^{-1}) = \cP^{-1},  
\end{equation}
\begin{equation}  
\sigma_l(P^{+0})= \cP^{+0}|_{\lambda \to \lambda^{-1}  } , \qquad 
\sigma_r(P^{-0} ) = \cP^{-0}|_{\lambda \to \lambda^{-1}  } .
\end{equation}
Note that $P^{\pm 1}$   are power series in negative powers of $\Lambda$ with leading term equal to $1$, while $P^{\pm 0}$ are power series in positive powers of $\Lambda$. 
We have that
\begin{align}
\label{Hirota3}
  \res_{\Lambda  }  \left[W^{+ 1}(q) \Lambda^{(N-1)n}W ^{- 1}(\bar q)\Lambda^k\right]=
\res_{\Lambda  } \left[W^{+ 0}(q) \Lambda^{-n}
W^{- 0}(\bar q)\Lambda^k\right],
\end{align}
where  $\bar q^N_0={q}^N_0-\frac{2\epsilon}{N} + \epsilon$.
Since this holds for $k\in\Z$ and there is no $k$ dependence in the square bracket we finally find the Hirota bilinear equation in difference operator form
\begin{equation}   \label{WWeq}
W^{+1}(q) \Lambda^{(N-1)n} W^{-1} (\bar{q})= W^{+0}(q) \Lambda^{-n}   W^{-0} (\bar{q}) .
\end{equation}

Now we proceed to derive some consequences from this bilinear equation. 
For $\bar q^\alpha_j={q}^\alpha_j-\delta_{\alpha N}\delta_{j0}(\frac{2\epsilon}{N}-\epsilon)$ we get
\begin{equation}  
P^{+1} \Lambda^{(N-1)n} P^{-1} = P^{+0} \Lambda^{-n}   P^{-0} , 
\end{equation}
which implies for $n=0$ and $b=0,1$ that 
\begin{equation}
\label{Pinverse}
P^{-b} (q^\alpha_j-\delta_{\alpha N}\delta_{j0}(\frac{2\epsilon}{N}-\epsilon))
= P^{+ b}(q^\alpha_j)^{-1}, 
\end{equation}
consequently for $n=1$ we obtain the constraint 
\begin{equation}  \label{dress-eth}
P^{+1} \Lambda^{N-1} (P^{+1})^{-1} = P^{+0} \Lambda^{-1}  (P^{+0})^{-1} =: L 
\end{equation}
where $L$ is a difference Lax operator of the form 
\begin{equation}
L= \Lambda^{N-1} +v_{N-2} \Lambda^{N-2}+v_{N-3} \Lambda^{N-3}+\cdots  v_0 + e^u \Lambda^{-1}.
\end{equation}
We define the following operators as in \cite{Carlet} or \cite{CvdL2013}:
\[
L^{\frac1{N-1}}= P^{+1} \Lambda (P^{+1})^{-1}
\]
and its logarithm 
\[
\log L= \frac12 \log_+ L+\frac12 \log_- L, \ \mbox{where } \log_+L=P^{+1}\epsilon\partial_x  (P^{+1})^{-1}, \ \log_-L=-P^{+0}\epsilon\partial_x  (P^{+0})^{-1}.
\]

\begin{remark}
For later use, we note that the coefficients $w_k$ of $\log L$  
\begin{equation}
\label{logLw}
\log L=\sum_{k\in\mathbb  Z} w_k \Lambda^k .
\end{equation}
are given by
\begin{equation}
\label{wk}\begin{aligned}
2w_{-k}&=  \epsilon\res_{\lambda=\infty }\lambda^{k-1}\cP^{+1}
\frac{\partial \cP^{-1}(x-\frac{2\epsilon }{N}+\epsilon(1- k))}{\partial x} d\lambda,\\
2w_{k}&=  -\epsilon\res_{\lambda =\infty}\lambda^{k-1}\cP^{+0}
\frac{\partial \cP^{-0}(x-\frac{2\epsilon }{N}+\epsilon(1+k))}{\partial x}  d\lambda.
\end{aligned}
\end{equation}
\end{remark}

If we differentiate \eqref{WWeq} by $q_\ell^\alpha$ and project on the negative, respectively non-negative,  degrees of $\Lambda$, we obtain the following Sato-Wilson equations
\begin{equation}
\label{eq:Sato-Wilson}
\frac{\partial P^{+1}}{\partial q_\ell^\alpha}=-(B_\ell^\alpha)_- P^{+1}, \quad 
\frac{\partial P^{+0}}{\partial q_\ell^\alpha}=(B_\ell^\alpha)_+P^{+0},
\end{equation}
where $B_\ell^\alpha$ is defined by 
\begin{equation}
\begin{aligned}
B_\ell^1  &= 
 \frac{N+2}{\epsilon 2N  (\ell+1)!}
  L^{\ell+1},\\
B_\ell^\alpha&= \frac 1 {\epsilon (N-1)\left(\frac{N-\alpha}{N-1}\right)_{(\ell+1)}}L^{\frac{N-\alpha}{N-1}+\ell},
\quad \alpha=2,\ldots N-1, 
\\
B_\ell^N  &= \frac1{\epsilon \ell !}\left( 2\log L-\frac N{N-1}\ch(\ell)\right)L^\ell,
\end{aligned}
\end{equation}
for $\ell \geq 0$. 
It is then straightforward to derive the Lax equations from the Sato-Wilson equations \eqref{eq:Sato-Wilson}:
\begin{equation} \label{lax-1}
\frac{\partial L}{\partial q_\ell^\alpha}= \left[\left(B_\ell^\alpha\right)_+, L\right]=
-\left[\left(B_\ell^\alpha\right)_-, L\right]
.
\end{equation}

\subsection{Lax representation with pseudo-differential operators}

In this section we derive the Lax equations in terms of pseudo-differential operators directly from the Hirota quadratic equations. For most of the times $q^\alpha_\ell$, those with $\alpha \not=N$, the derivation follows the usual procedure as in the case of the rational reductions of the KP hierarchy, but it is more complicated in the case of the ``logarithmic'' times $q^N_\ell$. 

As before, observe that the $\partial_x$-valued Hirota  one-form $\tilde \omega_\infty$ is regular at $\lambda=\infty$  iff the $\partial_x$-valued Hirota one-form 
\begin{equation} \label{omegahat}
\hat \omega_\infty :=
 e^{-
\sum_{\ell=1}^\infty \frac{\lambda^\ell}{\ell!} \frac{\ch(\ell)}{\epsilon}   q^N_\ell
+\frac1{2\epsilon}\sum_{\ell=0}^\infty
\frac{\lambda^{\ell+1}}{(\ell+1)!} q^1_\ell 
}\otimes
e^{
\sum_{\ell=1}^\infty \frac{\lambda^\ell}{\ell!} \frac{\ch(\ell)}{\epsilon}  q^N_\ell 
-\frac1{2\epsilon}\sum_{\ell=0}^\infty
\frac{\lambda^{\ell+1}}{(\ell+1)!}   q^1_\ell
}
\ \tilde \omega_\infty
\end{equation}
is regular. 
The Hirota equations are therefore satisfied iff 
\begin{equation}
\res_{\lambda=\infty} \hat\omega_\infty  \lambda^n d\lambda = 0 \text{ for } n \geq 0.
\end {equation} 

With a reasoning similar to the one used at the beginning of Subsection \ref{subsec:LaxDifference} we can rewrite these Hirota equations as  

\begin{align}
\label{Hir3} \notag
\res_{\lambda = \infty} \lambda^{(N-1)n+k}  
e^{-
\sum_{\ell=1}^\infty \frac{\lambda^{(N-1)\ell}}{\ell!} \frac{\ch(\ell)}{\epsilon}   (q^N_\ell - \bar{q}^N_\ell)
+\frac1{2\epsilon}\sum_{\ell=0}^\infty
\frac{\lambda^{(\ell+1)(N-1)}}{(\ell+1)!} (q^1_\ell - \bar{q}_\ell^1 )
}
\cW^{+ 1}(q)\cW ^{- 1}(\bar q) \frac{d \lambda}{\lambda} = \\
=
 \res_{\lambda = \infty} \lambda^{n-k}
 e^{-
\sum_{\ell=1}^\infty \frac{\lambda^\ell}{\ell!} \frac{\ch(\ell)}{\epsilon}   (q^N_\ell - \bar{q}^N_\ell)
+\frac1{2\epsilon}\sum_{\ell=0}^\infty
\frac{\lambda^{\ell+1}}{(\ell+1)!} (q^1_\ell - \bar{q}_\ell^1 )
}
 \cW^{+ 0}(q)\cW^{- 0}(\bar q)  \frac{d \lambda}{\lambda},
\end{align}
where  $q^N_0-\bar q^N_0 = \epsilon (k-1+\frac 2N)$, for $k \in \Z$ and $n\geq0$.
Defining
\begin{equation}
\begin{aligned}
\tilde\cE (q,\lambda):=\exp\Bigg( 
& \frac 1\epsilon \sum_{\alpha=2}^{N-1} 
\sum_{\ell=0}^\infty  \frac{\lambda^{{N-\alpha} 
 +\ell(N-1)}}{(N-1)\left(\frac{N-\alpha}{N-1}\right)_{(\ell+1)}}  
q^\alpha_\ell  +
\\
&+\frac{N+2}{2N \epsilon} \sum_{\ell=0}^\infty \frac{\lambda^{(N-1)(\ell+1)}}{(\ell+1)!} q^1_\ell    
 -  \sum_{\ell=1}^\infty \frac{N\ch(\ell)}{(N-1)\epsilon} \frac{\lambda^{\ell(N-1)}}{\ell!}  q^N_\ell \Bigg)
\end{aligned}
\end{equation}
the equation~\eqref{Hir3} becomes 
\begin{align} \notag
\res_{\lambda = \infty} \lambda^{(N-1)n+k}  
\cP^{+1}(q) \tilde{\cE}(q-\bar{q},\lambda) 
e^{\sum_{\ell=1}^\infty \frac{\lambda^{(N-1)\ell}}{\ell !} ( q^N_\ell- \bar q^N_\ell)\partial_x}
\cP^{-1}(\bar{q})
\frac{d \lambda}{\lambda} = \\
=
\res_{\lambda = \infty} \lambda^{n-k}
 \cP^{+0} (q)
e^{\sum_{\ell=1}^\infty \frac{\lambda^\ell}{\ell !} ( q^N_\ell- \bar q^N_\ell)\partial_x}
 \cP^{-0}(\bar{q})
\frac{d \lambda}{\lambda} .
\label{8.35}
\end{align}

Finally we introduce a new ``spatial'' variable by replacing $q_0^{N-1}$ by $q_0^{N-1}+X$ and denote
\begin{equation}  
\label{new8.9}
\tilde\cP^{\pm 1}
=\cP^{\pm 1}|_{q_0^{N-1}\to q_0^{N-1}+X},
\qquad 
\tilde\cP^{\pm 0}
=\cP^{\pm 0}|_{q_0^{N-1}\to q_0^{N-1}+X}
 .
\end{equation}

\begin{lemma}
The Hirota equation~\eqref{Hirota} is equivalent to the following equality of residues for $k\in \mathbb Z$ and $n\geq 0$:
\begin{equation}
\begin{aligned}
\label{Hirota4}
  &\res_{\lambda = \infty} \lambda^{n(N-1)+k} \tilde\cE (q-\bar q,\lambda ) \left(\tilde\cP^{+ 1}(q,\lambda )
e^{ 
\sum_{\ell=1}^\infty \frac{\lambda^{\ell(N-1)}}{\ell !} ( q^N_\ell- \bar q^N_\ell)\partial _x
}e^{-\epsilon k\partial_x}
\tilde\cP ^{- 1}(\bar q,\lambda )\right) e^{\frac{\lambda }\epsilon(X-\bar X)} {d\lambda} \\ 
&\qquad =\res_{\lambda = \infty} \lambda^{n-k-2}\left(\tilde\cP^{+ 0}(q,\lambda)
e^{
\sum_{\ell=1}^\infty \frac{\lambda^\ell}{\ell !} ( q^N_\ell- \bar q^N_\ell)\partial _x
}e^{-\epsilon k\partial_x}
\tilde\cP^{- 0}(\bar q,\lambda)\right){d\lambda} ,
\end{aligned}
\end{equation}
for $\bar q^N_0= q^N_0  -\frac{2\epsilon}N $.
\end{lemma}

Notice that in the previous equation the variable $\overline q_0^{N-1}$ was clearly shifted by $\overline{X}$. Moreover, to simplify the expression, we set $\bar{q}_0^N \mapsto \bar{q}_0^N -\epsilon(k-1)$, replaced $k$ by $k+1$,
and finally multiplied~\eqref{8.35} on the right by $e^{- \epsilon k \partial_x}$. 

This bilinear identity is equivalent to a bilinear identity involving pseudo-differential operators in $X$. This can be seen by using the fundamental lemma, see Lemma~50 in Section 8.3.2 in~\cite{carletHigherGeneraCatalan2021}. 
\begin{proposition}
The Hirota equation~\eqref{Hirota} is equivalent to
\begin{equation}
\begin{aligned}
\label{Hirota5} 
&\Big[  \tilde\cP^{+ 1}(q, \epsilon\partial_X )\tilde\cE (q-\bar q, \epsilon\partial_X ) 
e^{ 
\sum_{\ell=1}^\infty \frac{(\epsilon\partial_X)^{\ell(N-1)}}{\ell !} ( q^N_\ell- \bar q^N_\ell)\partial _x
}e^{-\epsilon k\partial_x}
(\epsilon\partial_X)^{n(N-1)+k}
\tilde\cP ^{- 1}(\bar q,-\epsilon\partial_X)^*\Big]_-  = \\ 
&\qquad =\res_{\lambda} \lambda^{n-k-2}\left(\tilde\cP^{+ 0}(q,\lambda)
e^{
\sum_{\ell=1}^\infty \frac{\lambda^\ell}{\ell !} ( q^N_\ell- \bar q^N_\ell)\partial _x
}e^{-\epsilon k\partial_x}(\epsilon\partial_X)^{-1}
\tilde\cP^{- 0}(\bar q,\lambda)\right)d\lambda
\end{aligned}
\end{equation}
for $k\in\mathbb Z$, $n\geq0$ and $\bar q^N_0= q^N_0  -\frac{2\epsilon}N$.
\end{proposition}
Notice that in the previous equation we have $\bar{X}=X$.  As usual the symbol $P(\lambda) = \sum_k p_k \lambda^k$ is quantised to an operator $P(\epsilon \partial_X)$ and its formal adjoint $P^*(\epsilon \partial_X)$ as
\begin{equation}
P(\epsilon \partial_X) = \sum_k p_k (\epsilon \partial_X)^k, \qquad 
P^*(\epsilon \partial_X) = \sum_{k} (-\epsilon \partial_X)^k p_k. 
\end{equation}
Moreover $[ \cdot ]_-$ denotes the projection to negative powers of $\epsilon \partial_X$. 

Let us spell out the main consequences of the Hirota equation~\eqref{Hirota5}. 
Let us a first set $q=\bar q$, except for the case $\bar q^N_0=q^N_0  -\frac{2\epsilon}N$.  We obtain
\begin{equation}
\begin{aligned}
\label{Hirota6} 
  &\Big[   \tilde\cP^{+ 1}(x,X,q, \epsilon\partial_X )
 e^{-\epsilon k\partial_x}
(\epsilon\partial_X)^{n(N-1)+k}
\tilde\cP ^{- 1}(x-\frac{2\epsilon}N ,X,
q,-\epsilon\partial_X)^*\Big]_-  =\\
 &\quad=\res_{\lambda } \lambda^{n-k-2}\left(\tilde\cP^{+ 0}(x,X,q,\lambda)
 e^{-\epsilon k\partial_x}(\epsilon\partial_X)^{-1}
\tilde\cP^{- 0}(x-\frac{2\epsilon}N ,X, q,\lambda)\right)d\lambda.
\end{aligned}
\end{equation}

Let us consider the case $n\leq k$. In such case the righthand side in~\eqref{Hirota6} vanishes since $\tilde{\cP}^{\pm 0}$ only contains non-positive powers of $\lambda$. For $n=0$ and $k=0$ we obtain that 
\begin{equation} \label{Pinvvv}
\tilde\cP ^{- 1}(x-\frac{2\epsilon}N,X,q,-\epsilon\partial_X)^*= \tilde\cP^{+ 1}(x,X,q, \epsilon\partial_X )^{-1},
\end{equation}
since  $\tilde\cP^{+1}(q, \epsilon\partial_X)$ is a power series in $(\epsilon \partial_X)^{-1}$ with leading term equal to $1$. Substituting this back, we obtain that~\eqref{Hirota6} for $n\leq k$ has the form
\begin{equation}
\label{Hirota6b}
\Big[   \tilde\cP^{+ 1}(x,X,q, \epsilon\partial_X )
 e^{-\epsilon k\partial_x}
(\epsilon\partial_X)^{n(N-1)+k}
\tilde\cP^{+ 1}(x,X,q, \epsilon\partial_X )^{-1}
\Big]_-  =0.
\end{equation}



The constraints encoded by this equations boil down to the two cases $n=0$ and $n=1$ for $k=1$. Indeed, defining
\begin{equation}
S(q,e^{\epsilon \partial_x},\epsilon \partial_X) =
\tilde\cP^{+ 1}(q, \epsilon\partial_X )
 e^{-\epsilon \partial_x}
(\epsilon\partial_X)
\tilde\cP ^{+ 1}(q,\epsilon\partial_X)^{-1}
\end{equation}
and
\begin{equation}
T(q,e^{\epsilon \partial_x},\epsilon \partial_X) :=
 \tilde\cP^{+ 1}(q, \epsilon\partial_X )
 e^{-\epsilon \partial_x}
(\epsilon\partial_X)^{N}
\tilde\cP ^{+1}(q,\epsilon\partial_X)^{-1}
\end{equation}
the equation~\eqref{Hirota6b} implies, for $n=0$, $n=1$ and $k=1$, that both operators are differential in $X$, i.e., they do not contain negative powers of $\epsilon \partial_X$. In particular we have
\begin{equation}
\label{S}
S(q,e^{\epsilon \partial_x},\epsilon \partial_X) =
\left(
\epsilon\partial_X e^{-\epsilon \partial_x}+\tilde\cP^{+ 1}(q)_{-1} e^{-\epsilon \partial_x}-e^{-\epsilon \partial_x}\tilde\cP^{+1}(q)_{-1}\right)
=\left(\epsilon\partial_X-\phi(q)\right)e^{-\epsilon \partial_x},
\end{equation}
where $\phi := (1- e^{-\epsilon \partial_x})( \tilde\cP^{+1}(q)_{-1})$ and $T(q,e^{\epsilon \partial_x},\epsilon \partial_X) = T(q,e^{\epsilon \partial_x},\epsilon \partial_X)_+$. Moreover, once these two constraints are satisfied, the remaining contraints in~\eqref{Hirota6b} are also implied, since they only state that $T^n S^{k-n}$ is differential. 

%
%

Let us now define the Lax operator as
\begin{equation}
 \cL (q,\epsilon\partial_X) := \tilde\cP^{+ 1}(q, \epsilon\partial_X )
(\epsilon\partial_X)^{N-1}
\tilde\cP^{+ 1}(q, \epsilon\partial_X )^{-1}.
\end{equation}
Its differential part is of the form
\begin{equation}
\cL (q,\epsilon\partial_X)_+= (\epsilon\partial_X)^{N-1} + a_{2} (\epsilon\partial_X)^{N-3} + \cdots + a_{N-1} .
\end{equation}
We also use the notation
$L(q, \epsilon\partial_X) := \tilde\cP^{+ 1}(q, \epsilon\partial_X )
\epsilon\partial_X
\tilde\cP^{+ 1}(q, \epsilon\partial_X )^{-1}$ so that $ \cL (q,\epsilon\partial_X) =L(q, \epsilon\partial_X )^{N-1}$. 
Notice that we have
\begin{equation}
\cL (q,\epsilon\partial_X) = T(q,e^{\epsilon \partial_x},\epsilon \partial_X)    S(q,e^{\epsilon \partial_x},\epsilon \partial_X)^{-1}
=    S(q,e^{\epsilon \partial_x},\epsilon \partial_X)^{-1} T(q,e^{\epsilon \partial_x},\epsilon \partial_X),
\end{equation}
since $T(q,e^{\epsilon \partial_x},\epsilon \partial_X)$ and $S(q,e^{\epsilon \partial_x},\epsilon \partial_X)$ commute. 
If we define the differential operator $\cT(q,\epsilon \partial_X)$  of order $N$ and $\cS(q,\epsilon \partial_X)$ of order $1$ by 
\begin{equation}
T(q,e^{\epsilon \partial_x},\epsilon \partial_X) = \cT(q,\epsilon \partial_X) e^{-\epsilon \partial_x} , \qquad 
S(q,e^{\epsilon \partial_x},\epsilon \partial_X) = \cS(q,\epsilon \partial_X) e^{-\epsilon \partial_x},
\end{equation}
i.e. $\cS(q,\epsilon \partial_X) =\epsilon\partial_X-\phi(q)$, we can also write
\begin{equation}
\cL (q,\epsilon\partial_X) = \cT(q,\epsilon \partial_X)    \cS(q,\epsilon \partial_X)^{-1}
=  e^{\epsilon \partial_x}  \cS(q,\epsilon \partial_X)^{-1} \cT(q,\epsilon \partial_X) e^{-\epsilon \partial_x}.
\end{equation}
Notice that $\cS(q,\epsilon \partial_X)$ and $ \cT(q,\epsilon \partial_X)$ do not commute. 

Taking into account the steps performed so far, equation~\eqref{Hirota6} becomes
\begin{equation}
\label{LSmin}
(\cL^n S^{k})_-=\res_{\lambda } \lambda^{n-k-2}\tilde\cP^{+ 0}(x,X,q,\lambda)e^{-\epsilon k\partial_x}(\epsilon\partial_X)^{-1} 
\tilde\cP^{- 0}(x-\frac{2\epsilon}N ,X, q,\lambda)
d\lambda 
\end{equation}
for $n\geq0$ and $k\in\Z$. As observed above, for $n \leq k$ both sides are trivial. 
For arbitrary fixed $k$, equation~\eqref{LSmin} implies that
\begin{equation}
( \sum_{n=0}^\infty \lambda^{-n}  \cL^n S^k )_-
=\left( \lambda^{-k-1} \tilde\cP^{+ 0}(x,X,q,\lambda)
e^{-\epsilon k \partial_x} 
(\epsilon\partial_X)^{-1}
\tilde\cP^{- 0}( x-\frac{2\epsilon}N,X,q,\lambda) \right)_{\lambda^{\leq 0}},
\end{equation}
which in particular for $k=0$ gives
\begin{equation}
\label{Lminusseries}
\sum_{n=0}^\infty \lambda^{-n}(\cL (q,\epsilon\partial_X)^{n+1})_- =
 \tilde\cP^{+ 0}(x,X,q,\lambda)
  (\epsilon\partial_X)^{-1}
\tilde\cP^{- 0}( x-\frac{2\epsilon}N,X,q,\lambda).
\end{equation}

For $n=1$ and $k=0$ equation~\eqref{LSmin} gives
\begin{equation}
\cL (q,\epsilon\partial_X)_- = \tilde\cP^{+ 0}(x ,X ,q)_0
(\epsilon\partial_X)^{-1}
\tilde\cP^{- 0}( x-\frac{2\epsilon}N,X,q)_0
\end{equation}
while for $n=0$ and $k=-1$ we get
\begin{equation}
S^{-1} = \tilde\cP^{+ 0}(x,X,q,\lambda)_0
e^{\epsilon  \partial_x} 
(\epsilon\partial_X)^{-1}
\tilde\cP^{- 0}( x-\frac{2\epsilon}N,X,q,\lambda)_0.
\end{equation}
From the last two equations and~\eqref{S} we see that
\begin{equation}
\cL (q,\epsilon\partial_X)_- = a_N (\epsilon \partial_X - a_1)^{-1}
\end{equation}
where $a_N = \tilde\cP^{+ 0}(x ,X ,q)_0 \tilde\cP^{+ 0}(x-\epsilon ,X ,q)_0^{-1}$ and $a_1 = \phi$. Similarly we can also write $\cL (q,\epsilon\partial_X)_- =  (\epsilon \partial_X - \tilde{a}_1)^{-1} \tilde{a}_N$ for $\tilde{a}_1 = \phi(x+\epsilon)$ and $\tilde{a}_N= \tilde\cP^{- 0}( x-\frac{2\epsilon}N+\epsilon,X,q)_0^{-1} \tilde\cP^{- 0}( x-\frac{2\epsilon}N,X,q)_0$.

%
%
%
%

The standard procedure followed so far yields the Lax representation for most of the times, excluding the times $q^N_\ell$ for $\ell\geq0$, and allows to identify the hierarchy with the rational $(N-1)$-constrained KP hierarchy.
\begin{proposition} \label{prop:RRKP}
The Lax operator $\cL$ is of the form 
\begin{equation}
\cL (q,\epsilon\partial_X)= (\epsilon\partial_X)^{N-1} + a_{2} (\epsilon\partial_X)^{N-3} + \cdots + a_{N-1} + a_N (\epsilon \partial_X - a_1)^{-1}
\end{equation}
and  satisfies the Lax equations
\begin{equation}
\label{L8.40}
\begin{aligned}
\epsilon \frac{\partial    \cL (q,\epsilon\partial_X)}{\partial q_\ell^\alpha}  &=\left[
\left(B_\ell^\alpha(q,\epsilon\partial_X) \right)_+, \cL(q,\epsilon\partial_X)\right]
\\
&=-\left[
\left(B_\ell^\alpha(q,\epsilon\partial_X) \right)_-, \cL(q,\epsilon\partial_X)\right]
,\quad 1\le \alpha<N,
 \end{aligned}
\end{equation}
where
\begin{equation}
\label{Boper1}
B_\ell^1(q,\epsilon \partial_X) := 
\frac{N+2}{2N }
  \frac{\cL(q,\epsilon\partial_X)^{ 
 \ell+1}}{ (\ell+1)!},\quad 
B_\ell^\alpha (q,\epsilon \partial_X):=
  \frac{L(q,\epsilon\partial_X)^{ {N-\alpha} 
 +\ell(N-1)}}{(N-1) \left(\frac{N-\alpha}{N-1}\right)_{(\ell+1)}}
,
\end{equation}
for $\quad 1<\alpha<N$.
\end{proposition}
\begin{proof}
We obtain the following Sato-Wilson equations if we differentiate \eqref{Hirota5} with respect to $q^\alpha_\ell$, for $\alpha\ne N$ and set $n=k=0$, noticing that the righthand side vanishes and using~\eqref{Pinvvv}
\begin{equation}
\label{8.40}
\epsilon \frac{\partial    \tilde\cP ^{+1}(q,\epsilon\partial_X)}{\partial q_\ell^\alpha} \tilde\cP ^{+1}(q,\epsilon\partial_X)^{-1}=-B_\ell^\alpha(q,\epsilon\partial_X)_-
   ,\quad 1\leq\alpha<N.
\end{equation}
From this equations \eqref{L8.40} immediately follows.
\end{proof}

In the literature the constrained KP hierarchy is also represented as in the following proposition. 

\begin{proposition}
The Lax operator $\cL$ is of the form
\begin{equation}
\label{Hirota8}
\begin{aligned}
\cL(q, \epsilon\partial_X   )=
\cL(q, \epsilon\partial_X  )_++
\tilde\cP^{+ 0}(x,X,q)_0
 (\epsilon\partial_X)^{-1}
\tilde\cP^{- 0}( x-\frac{2\epsilon}N,X,q)_0,
\end{aligned}
\end{equation}
where the $\tilde\cP^{\pm 0}(x-\frac{\epsilon}N\pm\frac{\epsilon}N ,X,q)_0$ satisfy
\begin{equation}
\label{eigen}
\begin{aligned}
\epsilon \frac{\partial    \tilde\cP^{+ 0}(x ,X,q)_0}{\partial q_\ell^\alpha}  &= 
\left(B_\ell^\alpha(q,\epsilon\partial_X) \right)_+ \left(
\tilde\cP^{+ 0}(x ,X,q)_0
\right),\\
 \epsilon \frac{\partial    \tilde\cP^{- 0}( x-\frac{2\epsilon}N,X,q)_0}{\partial q_\ell^\alpha}  &= -
\left(B_\ell^\alpha(q,\epsilon\partial_X)^* \right)_+ \left(
\tilde\cP^{- 0}( x-\frac{2\epsilon}N,X,q)_0
\right)
,\quad 1\le \alpha<N.
\end{aligned}
\end{equation}
\end{proposition}
\begin{proof}
If we differentiate \eqref{Hirota5}, for $n=1$, $k=0$,  by $q^\alpha_j$ and use \eqref{8.40}, we obtain that 
\[
\begin{aligned}
\epsilon \frac{\partial \tilde\cP^{+ 0}(x,X,q)_0}{\partial q_\ell^\alpha}&
 (\epsilon\partial_X)^{-1}
\tilde\cP^{- 0}( x-\frac{2\epsilon}N,X,q)_0=\\
=&\left(\epsilon
\frac{\partial    \tilde\cP ^{+1}(q,\epsilon\partial_X)}{\partial q_\ell^\alpha} 
 (\epsilon\partial_X)^{N-1}
\tilde\cP ^{+1}(q,\epsilon\partial_X)^{-1}
+B_\ell^\alpha (q,\epsilon\partial_X)\cL(q,\epsilon\partial_X)
\right)_-\\
=&
\left(-
 B_\ell^\alpha (q,\epsilon\partial_X)_-\cL(q,\epsilon\partial_X)
+B_\ell^\alpha (q,\epsilon\partial_X)\cL(q,\epsilon\partial_X)
\right)_-\\
=&
\left( 
B_\ell^\alpha (q,\epsilon\partial_X)_+\cL(q,\epsilon\partial_X)_-
\right)_-\\
=&\left(B_\ell^\alpha (q,\epsilon\partial_X)_+
\tilde\cP^{+ 0}(x,X,q)_0 
 (\epsilon\partial_X)^{-1}
\tilde\cP^{- 0}( x-\frac{2\epsilon}N,X,q)_0\right)_-\\
=&B_\ell^\alpha (q,\epsilon\partial_X)_+
\left(\tilde\cP^{+ 0}(x,X,q)_0 \right)
 (\epsilon\partial_X)^{-1}
\tilde\cP^{- 0}( x-\frac{2\epsilon}N,X,q)_0.
\end{aligned}
\]
wich gives the first equation of \eqref{eigen}. The second equation can be obtained by differentiating  \eqref{Hirota5},   for $n=1$, $k=0$,   by $\bar q^\alpha_j$ and using the adjoint version of the Sato-Wilson equations \eqref{8.40}.
\end{proof}

Using the Sato-Wilson equations~\eqref{8.40},  we obtain the following version of Krichever's rational reduction of KP, cf. \cite[Theorem 2]{Kr1995} or \cite{HvdL}:
\begin{proposition} \label{prop:KricheverRR}
The  differential operators $\cT$ and $\cS$ satisfy
\begin{equation}
\label{Krichever}
\begin{aligned}
\epsilon m_{\ell}^\alpha \frac{ \partial \cT }{\partial q_\ell^\alpha} &=
\left(\left(\cT \cS^{-1}\right)^{\frac{ {N-\alpha} 
 +\ell(N-1)}{N-1}}\right)_+\cT -
\cT\left(\left(\cS^{-1}\cT\right)^{\frac{ {N-\alpha} 
 +\ell(N-1)}{N-1}}\right)_+,\\
\epsilon m_{\ell}^\alpha \frac{ \partial \cS }{\partial q_\ell^\alpha} &=
\left(\left(\cT \cS^{-1}\right)^{\frac{ {N-\alpha} 
 +\ell(N-1)}{N-1}}\right)_+\cS -
\cS\left(\left(\cS^{-1}\cT\right)^{\frac{ {N-\alpha} 
 +\ell(N-1)}{N-1}}\right)_+,
\end{aligned}
\end{equation}
where 
\[
m_\ell^1=\frac{N+2}{(\ell+1)!2N },\quad 
m_\ell^\alpha
=  \frac1{(N-1) \left(\frac{N-\alpha}{N-1}\right)_{(\ell+1)}}
,\quad 1<\alpha<N.
\]
\end{proposition}

\bigskip

We now want to find the Lax representation for the ``logarithmic'' flows corresponding to the times $q_\ell^N$ for $\ell >0$. 
We differentiate~\eqref{Hirota5} with respect to $q^N_\ell$ and put $k=n=0$ and $q=\bar q$,  except for $\bar q^N_0=q^N_0  -\frac{2\epsilon}N$, obtaining:
\begin{equation}
\begin{aligned}
\epsilon \frac{\partial \tilde\cP ^{+1}(q,\epsilon\partial_X)}{\partial q_\ell^N}
&
\tilde\cP ^{+1}(q,\epsilon\partial_X)^{-1}=-
\Big( \frac{(\cL(q,\epsilon\partial_X)^\ell}{\ell !}\Big(\tilde\cP ^{+1}(q,\epsilon\partial_X)\epsilon \partial_x 
\tilde\cP ^{+1}(q,\epsilon\partial_X)^{-1} -\frac{N\ch(\ell)}{N-1 }  \Big)\Big)_- \\
&+\res_{\lambda } \frac{\lambda^{\ell-2}}{\ell !}\left(\tilde\cP^{+ 0}(x,X,q,\lambda)
 (\epsilon\partial_X)^{-1} \epsilon\partial_x
\tilde\cP^{- 0}(x-\frac{2\epsilon}N,X, q,\lambda) \right)d\lambda\\
=&
\Big( \frac{(\cL(q,\epsilon\partial_X)^\ell}{\ell !}\Big(\epsilon\frac{\partial \tilde\cP ^{+1}(q,\epsilon\partial_X) }{\partial x} \tilde\cP ^{+1}(q,\epsilon\partial_X)^{-1}  
+\frac{N\ch(\ell)}{N-1 }  
\Big)\Big)_-
+\\
&\qquad+\epsilon\res_{\lambda} \frac{\lambda^{\ell-2}}{\ell !}\left(\tilde\cP^{+ 0}(x,X,q,\lambda)
 (\epsilon\partial_X)^{-1}  
\frac{\partial \tilde\cP^{- 0}(x-\frac{2\epsilon}N,X, q,\lambda)}{\partial x} \right)d\lambda .
\end{aligned}
\end{equation}
Notice that here we have used~\eqref{Lminusseries} to obtain the second expression which does not contain the operator $\epsilon \partial_x$. 
Using again~\eqref{Lminusseries},
we find that the last term on the right-hand side is equal to 
\begin{equation}
\begin{aligned}
\label{8.42}
 \epsilon\res_{\lambda} \sum_{k=0}^{\ell-1} \frac{\lambda^{k-1}}{\ell !}
(\cL^{\ell-k})_- 
\tilde\cP^{- 0}( x-\frac{2\epsilon}N,X,q,\lambda)^{-1}\frac{\partial \tilde\cP^{- 0}(x-\frac{2\epsilon}N,X, q,\lambda)}{\partial x} d\lambda
\end{aligned}
\end{equation}
where the sum is finite since the terms with $k<0$ involve only powers of $\lambda$ strictly smaller than $-1$.

 If we define the operator $\cQ(q,e^{\epsilon \partial_x}, \epsilon\partial_X)$ as the dressing of $e^{\epsilon \partial_x}$, i.e. 
 \begin{equation}
 	\cQ = \tilde\cP ^{+1}(q,\epsilon\partial_X)e^{ \epsilon \partial_x }\tilde\cP ^{+1}(q,\epsilon\partial_X)^{-1},
 \end{equation} for which in particular $L = \cQ S$  and $L^N = \cQ T$, then it is quite natural to denote with $\log\cQ(q,\epsilon \partial_x, \epsilon\partial_X)$  the dressing of the operator $\epsilon \partial_x$:
\begin{equation}
\log\cQ=\tilde\cP ^{+1}\epsilon\partial_x (\tilde\cP ^{+1})^{-1}
= \epsilon\partial_x-
\epsilon\frac{\partial \tilde\cP ^{+1} }{\partial x}( \tilde\cP ^{+1}  )^{-1}
\end{equation}
which allows us to state the Sato-Wilson equations for the ``logarithmic'' flows as in the following
\begin{proposition}
\label{logprop}
The wave operator $\tilde\cP ^{+1}(q,\epsilon\partial_X)$ satisfies the Sato-Wilson equations  
\begin{equation}
\epsilon \frac{\partial \tilde\cP ^{+1}(q,\epsilon\partial_X)}{\partial q_\ell^N} \tilde\cP ^{+1}(q,\epsilon\partial_X)^{-1}= 
-(B_\ell^N)_-
\end{equation}
where
\begin{equation}
\label{Boper2}
B_\ell^N = 
\frac{\cL(q,\epsilon\partial_X)^\ell}{\ell !} \left( \log\cQ - \epsilon\partial_x   
-\frac{N\ch(\ell)}{N-1 }  -\tilde{B}_\ell \right)
\end{equation}
and the operator $\tilde{B}_\ell$ is defined by
\begin{equation} \label{tildeB}
\tilde{B}_{\ell} = \epsilon\sum_{k=0}^{\ell-1} \cL^{-k} \res_\lambda \lambda^k \tilde\cP^{- 0}( x-\frac{2\epsilon}N,X,q,\lambda)^{-1}\frac{\partial \tilde\cP^{- 0}(x-\frac{2\epsilon}N,X, q,\lambda)}{\partial x} \frac{d\lambda}{\lambda}.
\end{equation}
The Lax equations for the logarithmic flows are then simply given by
\begin{equation}
\label{8.47L}
\epsilon \frac{\partial  \cL(q,\epsilon\partial_X)}{\partial q_\ell^N}
 =\left[ -(B_\ell^N)_-, 
\cL(q,\epsilon\partial_X)\right].
\end{equation}
\end{proposition}

Notice that the operator $\tilde{B}_\ell$ can be written as 
\begin{equation}
\tilde{B}_\ell =   \epsilon\sum_{k=0}^{\ell-1} \cL^{-k} \frac{\partial \tilde{v}_k}{\partial x} 
\end{equation}
where the $\tilde{v}_k$ are the coefficients in the logarithm of $\tilde\cP^{- 0}( x-\frac{2\epsilon}N,X,q,\lambda)$, i.e.
\begin{equation}
\log \tilde\cP^{- 0}( x-\frac{2\epsilon}N,X,q,\lambda) = \sum_{k\geq 0} \tilde{v}_k  \lambda^{-k}.
\end{equation}

\begin{remark}
While $\log \cQ$ clearly commutes with $\cL$, since it is defined by dressing $\epsilon \partial_x$, we cannot claim that the operator $\tilde{B}_\ell$ commutes with $\cL$. 
\end{remark}

\begin{remark}
Notice that in the $N=2$ case discussed in~\cite{carletHigherGeneraCatalan2021} we were able to express $B_\ell^N$ in terms of a single logarithm operator $\log \cL$ without the more complicated part $\tilde{B}_\ell$ given in terms of residues in~\eqref{tildeB}. This somehow ad-hoc definition of $\log \cL$ relied on a certain symmetry of the $N=2$ Lax operator that allowed to express the coefficients $w_k$ with positive $k$ in terms of those with negative $k$. This symmetry is however not present in the general $N>2$ case. 
\end{remark}
\begin{remark}
\label{extraconj}
The importance of the Lax formulation of a hierarchy lies in the fact that it gives an explicit representation of its flows, namely it provides a way to write the time derivatives of the dependent variables appearing in the Lax operator $\cL$, the variables $a_1, \dots , a_N$ in this case, as differential polynomials in the same variables and their $X$-derivatives. In general this follows from the fact the the operators $B_\ell^\alpha$ defined in~\eqref{Boper1} and~\eqref{Boper2} are functions of $\cL$, namely that they commute with $\cL$ and that their coefficients are uniquely determined as $X$ differential polynomials in the variables $a_i$. While this is trivially true for the operators $B_\ell^\alpha$ for $1\leq \alpha <N$ and $\ell \geq0$, this cannot be easily proved for the operators $B_\ell^N$. Here we conjecture in particular that the operators $\log \cQ$ and $\tilde{B}_\ell$ satisfy such properties. 
\end{remark}

\begin{remark}
We can relate    $\frac{\partial \tilde{v}_k}{\partial x}$  and $\log\cQ$   
   to  the coefficients of the logarithm defined in~\S\ref{subsec:LaxDifference}:

 \begin{lemma}
Let $\tilde w_k=w_k|_{q_0^{N-1}\to q_0^{N-1}+X}$,  where $w_k$ is given by \eqref{wk},   then 
\begin{equation}
\label{vw}
 \frac{\partial \tilde{v}_k}{\partial x}  = -\frac 2\epsilon   \tilde w_k(x-(k+1)\epsilon)
\end{equation}
and
\begin{equation}
\label{lllog+}
\tilde \log\cQ = 2\sum_{k<0}\tilde w_k \Lambda^kS^k.
\end{equation}
%
\end{lemma}
\begin{proof}
We first   prove \eqref{vw}. 
For simplicity we will identify $\Lambda=e^{\epsilon \partial _x}$.
Using \eqref{Pinverse},  we have that $P^{+0}(\Lambda)P^{-0}(x-\frac{2\epsilon}N+\epsilon, \Lambda)=1$.  Now writing $P^{+0}(\Lambda)=\sum_{i=0}^\infty P^{+0}_j\Lambda^j$, and
$P^{-0}(x-\frac{2\epsilon}N+\epsilon, \Lambda)=\sum_{j=1}^\infty \Lambda^jP^{-0}_j(x)$, this means that for $k=0,1,\ldots$,  we have
$
\sum_{j=0}^k P^{+0}_j \Lambda^k P^{-0}_{k-j}(x)=\delta_{k0}
$,  or stated differently,
\[
\sum_{j=0}^k \Lambda^{-k-1}P^{+0}_j \Lambda^{k+1}P^{-0}_{k-j}(x-\epsilon)\lambda^{-k}=
\delta_{k0}.
\]
Thus 
\[
\begin{aligned}
\res_{\mu } \sum_{k=0}^{\infty }  
\mu^{k-1}\Lambda^{-k-1} 
\sum_{j=0}^k  P^{+0}_j \Lambda^{k+1}P^{-0}_{k-j}(x-\epsilon)\mu^{-k}
 d\mu \lambda^{-k}
&=\\
=\res_{\mu } \sum_{k=0}^{\infty }  
\mu^{k-1}\Lambda^{-k-1} 
\cP^{+0}(\mu)      \Lambda^{k+1}    \cP^{-0}(x-\frac{2\epsilon}N , \mu)d\mu
\lambda^{-k}
&=1
\end{aligned}
\]
Now replacing $q_0^{N-1}$ by $q_0^{N-1}+X$,  we get the same formula but now for 
$\tilde \cP^{\pm0}$ instead of $\cP^{\pm0}$.  Thus we now know how to calculate the action of the  inverse of $\tilde\cP^{- 0}( x-\frac{2\epsilon}N, \lambda)$ on 
$\frac{\partial \tilde\cP^{- 0}(x-\frac{2\epsilon}N, \lambda)}{\partial x}$, viz. as above:
\[
\tilde\cP^{- 0}( x-\frac{2\epsilon}N, \lambda)^{-1}\frac{\partial \tilde\cP^{- 0}(x-\frac{2\epsilon}N, \lambda)}{\partial x}=
\res_{\mu  } \sum_{k=0}^{\infty }  
\mu^{k-1}\Lambda^{-k-1} 
\cP^{+0}(\mu)      \Lambda^{k+1}   \frac{\partial  \cP^{-0}(x-\frac{2\epsilon}N , \mu)}{\partial x}d\mu
\lambda^{-k}
\]
Finally  using the second formula of 
\eqref{wk}, where we also replace $q_0^{N-1}$ by $q_0^{N-1}+X$, we obtain the desired result.

We will proof formula \eqref{lllog+} below,   see  \eqref{llog+}.
\end{proof}

%

Recall that the  variable $X$ was introduced by the shift $q_0^{N-1} \mapsto q_0^{N-1}+ X$.  The Lax equation for the time $q_0^{N-1}$ reads:

\begin{equation}
\epsilon\frac{\partial L}{\partial X}=
\epsilon\frac{\partial L}{\partial q^{N-1}_0}=[(L^{\frac1{N-1}})_+,L] =[\Lambda +f,L],
\end{equation}
with $f=(1-\Lambda) (1-\Lambda^{N-1})^{-1}v_{N-2}$, which is explicitly given by 
\begin{equation}
\epsilon(v_i)_X =(\Lambda-1)  (v_{i-1} )
+v_i  (1-\Lambda^i) (1-\Lambda^{N-1})^{-1}(1-\Lambda)(v_{N-2})
\end{equation}
for $-1 \leq i \leq N-2$, where we noted $v_{-1}=e^u$ and $v_{N-1}=1$. 

In principle we can invert these formulas as in the extended NLS case, see~\cite{Carlet2004Extended, carletHigherGeneraCatalan2021}, expressing the $x$ derivatives of the dependent variables $v_i$ as differential polynomials in $X$. Substituting in the equations~\eqref{lax-1} we obtain in a different way the   extended constrained KP hierarchy.

We want to replace the shift operator $\Lambda$ by $\epsilon\partial_X$, i.e.
 we want to associate to a dressing operator  $P^{+1}(\Lambda)= \sum_{i=0}^\infty p_i\Lambda^{-i}$ the following pseudo-differential operator
\[
\tilde\cP(\epsilon\partial_X)=\sum_{i=0}^\infty p_i (\epsilon\partial_X)^{-i}.
\]
Recall the definition of $S$ from \eqref{S}:
\[
S= \tilde\cP^{+1} (\epsilon\partial_X)\Lambda^{-1}(\tilde\cP^{+1})^{-1},
\]

The equation $LP^{+1}= P^{+1}\Lambda^{N-1}$  can be rewritten as
\[
\begin{aligned}
&\tilde\cP^{+1}(\epsilon\partial_X)^{N-1}
  =\\
&\tilde\cP^{+1}(x+(N-1)\epsilon)(\epsilon\partial_X)^{N-1} +v_{N-2}\tilde\cP^{+1}(x+(N-2)\epsilon) (\epsilon\partial_X)^{N-2} \cdots    + e^u  \tilde \cP^{+1}(x-\epsilon) (\epsilon\partial_X)^{-1}
\end{aligned}
\]
Now multiply with $(\tilde\cP^{+1})^{-1} $   from the right,  this gives 
\[
\begin{aligned}
&\tilde\cL= \Lambda^{N-1}S^{N-1} +v_{N-2}\Lambda^{N-2}S ^{N-2} \cdots    + e^u   \Lambda^{-1}S^{-1}.
\end{aligned}
\]
Similarly, the Sato-Wilson equation  ($j=1,\ldots, N$,  $n\ge  \delta_{j,N}$)
\[
\frac{\partial P^{+1}}{\partial q_n^j}=-(B_n^j)_- P^{+1}= -\sum_{k<0}a_{\ell,k}^j \Lambda^{k} P^{+1},
\]
turns into
\[
\frac{\partial \tilde\cP^{+1}}{\partial q_n^j}
(\tilde\cP^{+1})^{-1}= -\sum_{k<0}a_{\ell,k}^j   \tilde\cP^{+1}(x+k\epsilon)(\epsilon\partial_X)^{k}(\tilde\cP^{+1})^{-1}
=-\sum_{k<0}a_{\ell,k}^j   \Lambda^k\cS^{k}.
\]
Next we consider the equation 
\[
 P^{+1}\epsilon\partial_x=  \log_+ L  P^{+1} =(\epsilon\partial_x+2\sum_{k<0}w_k\Lambda^k )P^{+1},
\]
which gives that  $\log \cQ=\tilde\cP^{+1}\epsilon\partial_x(\tilde\cP^{+1})^{-1}$ is equal to 
\begin{equation}
\label{llog+}
\log\cQ= \epsilon\partial_x +2\sum_{k<0}w_k \Lambda^kS^k
\end{equation}
\end{remark}

\bibliographystyle{amsalpha} 
\bibliography{hypermaps}

\end{document}